\PassOptionsToPackage{dvipsnames}{xcolor}
\documentclass[acmsmall]{acmart}
\AtBeginDocument{%
  }

\setcopyright{acmlicensed}
\copyrightyear{2024}
\acmMonth{02}
\acmYear{2024}
\acmDOI{arXiv}
\acmConference[arXiv]{arXiv Preprint}{2024}{2024}
\acmISBN{}

\usepackage{subfig}
\usepackage{graphicx}
\usepackage{mathtools}
\usepackage[ruled]{algorithm2e}
\usepackage[dvipsnames]{xcolor}
\usepackage{multirow}
\usepackage{lipsum}                     
\usepackage{xargs}                      
\usepackage{todonotes}
\usepackage{url}

\usepackage{breakurl}

\begin{document}

\title{Honeybee: Byzantine Tolerant Decentralized Peer Sampling with Verifiable Random Walks}


\author{Yunqi Zhang}
\affiliation{%
  \institution{The Ohio State University}
  \city{Columbus}
  \country{USA}}
\email{zhang.8678@osu.edu}
\orcid{0000-0002-3363-9834}

\author{Shaileshh Bojja Venkatakrishnan}
\affiliation{%
  \institution{The Ohio State University}
  \city{Columbus}
  \country{USA}}
\email{bojjavenkatakrishnan.2@osu.edu}
\orcid{0000-0001-7355-634X}






\newcommand{\OA}{Honeybee}
\newtheorem{property}{Property}

\begin{abstract}
Popular blockchains today have hundreds of thousands of nodes and need to be able to support sophisticated scaling solutions---such as sharding, data availability sampling, and layer-2 methods. 
Designing secure and efficient peer-to-peer (p2p) networking protocols at these scales to support the tight demands of the upper layer crypto-economic primitives is a highly non-trivial endeavor. 
We identify decentralized, uniform random sampling of nodes as a fundamental capability necessary for building robust p2p networks in emerging blockchain networks. 
Sampling algorithms used in practice today (primarily for address discovery) rely on either distributed hash tables (e.g., Kademlia) or sharing addresses with neighbors (e.g., GossipSub), and are not secure in a Sybil setting. 
We present Honeybee, a decentralized algorithm for sampling nodes that uses verifiable random walks and table consistency checks. 
Honeybee is secure against attacks even in the presence of an overwhelming number of Byzantine nodes (e.g., \(\geq50\)\% of the network). 
We evaluate Honeybee through experiments and show that the quality of sampling achieved by Honeybee is significantly better compared to the state-of-the-art. Our proposed algorithm has implications for network design in both full nodes and light nodes.
\end{abstract}




\maketitle

\section{Introduction}\label{sec:intro}
Blockchains are steadily maturing into ecosystems that support decentralized applications (dapp) in diverse domains---finance, payments, storage, games, healthcare etc.---and used by millions of clients for conducting billions of dollars of transactions each day~\cite{coinmarketcap}. As demand for dapp usage increases, it is important that the blockchains can handle the high rate of transaction requests. 
Today major blockchains such as Bitcoin and Ethereum can process at most a few ten transactions per second on average which has resulted in unacceptably high transaction fees during periods of high demand~\cite{ethertxnfee}. 

To improve the scalability of blockchains, solutions such as state sharding, data sharding, and layer-2 methods (e.g., rollups) have been proposed. 
As these solutions are integrated into mainstream blockchains (e.g., Ethereum, Cardano, etc.), we are also seeing a big surge in the scale of the blockchain networks. 
For instance, today the Ethereum blockchain contains close to a million validator nodes.  
At these scales, designing secure and effective peer-to-peer networking protocols to support the complex upper-layer crypto-economic primitives is a significant challenge. 
Blockchains today are overtly reliant on random graph topologies, distributed hash tables (notably Kademlia), and p2p publish-subscribe protocols (such as GossipSub) for their networking needs. 
However, even the most basic problem of how to design a p2p network that eliminates eclipse attacks for a proof-of-work or proof-of-stake chain has received little attention in the literature. 
While heuristic approaches have offered an acceptable network performance historically, a growing network scale together with complex messaging primitives (beyond just broadcast) creates an urgent need to research principled approaches for p2p network design. 

In particular, we consider the ability of a node to sample peers uniformly at random from the network without reliance on centralized entities---especially when the network contains Byzantine or Sybil nodes---as a crucial capability necessary for emerging large-scale blockchain networks. 
A decentralized peer-sampling service provides several benefits. 
It enables peers to construct (and maintain) a robust overlay that is secure against eclipsing and network partitioning attacks. 
Even under a high fraction of dishonest nodes (e.g., 50\%), by randomly sampling and connecting to peers, honest nodes can ensure they form a connected subgraph. 
A tightly connected subgraph of honest nodes helps provide strong safety and liveness guarantees for the consensus algorithm running on the network. 
E.g., a ``balancing attack'' in proof-of-stake Ethereum becomes possible when adversaries form connections such that honest nodes have a lower network latency to adversarial nodes than to other honest nodes~\cite{neuder2021low,schwarz2022three,neu2022two}. 
This allows attackers to carefully time and release votes and blocks to honest nodes slowing down consensus.
A balancing attack is also possible if the adversary can obtain accurate statistics on the message propagation latencies in the network. 
If peers continuously obtain fresh address samples and regularly update the overlay, such a strategy will not work. 

A random peer sampling service also helps with sharding and data availability sampling. 
E.g., to address the scaling problem, the Ethereum community is focusing on developing a data-sharding design in which (1) block sizes are increased from the current 80 kB (average) to as big as 30 MB~\cite{dankshardingblocksize} and erasure coded, (2) each validator node stores only a small chunk of each coded block. 
Unlike Ethereum 1.0 which used a simple broadcast gossip primitive to disseminate transactions and blocks in the network or Ethereum 2.0 which uses a publish-subscribe model to disseminate different types of messages (e.g., attestations, blocks, transactions) over different subnets, a node in a data-sharded Ethereum must additionally be able to check the availability of a full block by sampling random chunks from other nodes~\cite{al2018fraud}. 
This requires a requesting node to be able to contact randomly selected nodes in the network and download chunks. Having the ability to sample random chunks is important not only for full functional nodes, but also for light nodes (e.g., a wallet running on a smart phone) which can number in the millions compared to the thousands of full nodes available today~\cite{ethmainnetstat}.
Peer sampling also helps blockchains employing full sharding. 
During a shard reconfiguration event, a migrating node must quickly discover nodes and download the latest state from the destination shard to minimize disruptions in the consensus. 
Last, peer sampling can help with server discovery in large-scale decentralized applications. 
E.g., in Livepeer a user selects a media server for encoding and distributing their video streams. 
A capability to select servers randomly not only helps with load balancing but can also add to security. 
The emerging trend of modular blockchains advocates for a separate data availability layer in the blockchain stack, the implementation of which can make use of peer sampling services as well~\cite{celestia,eigenda}. 



In an open, permissionless and decentralized setting that blockchains operate on, developing an algorithm by which a node can uniformly sample other nodes in the p2p network is far from straightforward. 
At present, the two common approaches for peer discovery are using (1) the Kademlia distributed hash table (DHT) protocol~\cite{maymounkov2002kademlia}, or (2) p2p exchange of addresses between neighbors in the overlay (e.g., used by GossipSub).  
To discover a random node in the network with Kademlia, all a node has to do is issue a query for a randomly selected target identifier and receive the IP address information of the node that is closest to the target identifier (in Kademlia’s XOR distance sense)~\cite{discv5,discv4}. 

In this paper we argue that in the presence of adversarial nodes (or Sybils), neither of the currently followed approaches achieve uniform sampling and can cause a requesting node to become eclipsed. 
As a potential replacement, in this paper we propose Honeybee which is a fully decentralized p2p algorithm for performing uniform node sampling even in the presence of adversarial nodes. A Honeybee node achieves sampling by participating in several random walks over the p2p overlay. Each node in Honeybee maintains an address table containing addresses of the most recently sampled peers which are also used to progress random walks of other peers visiting the node. To protect against adversarial attacks, Honeybee uses a verifiable public source of  randomness (e.g., can be derived from the blockchain) 
to perform the walks. Additionally, Honeybee nodes also perform peer-to-peer reconciliation of node address tables to identify and expose attackers engaging in equivocating their address tables.  

Sampling nodes in a distributed system using random walks has historically been well-studied in the context of applications such as overlay monitoring, design of expanders, search, routing, resource management etc~\cite{awan2006distributed}. However, prior works in this space consider models that do not simultaneously satisfy our requirements: (1) the algorithm must be decentralized, (2) there can be a large number (e.g., constant fraction) of adversarial nodes, and (3) adversarial nodes can exhibit arbitrary Byzantine behavior (e.g., message insertions, deletions during gossip). E.g., Anceaume et al.~\cite{anceaume2013power,anceaume2013uniform} consider achieving sampling through streaming messages between neighbors, but assume there are no Sybil attacks. Augustine et al.~\cite{augustine2015enabling} consider a dynamic p2p network model where the attacker decides how the network churns from round to round. However, within each round the assumption is that gossiping happens without any message loss. The works Awan et al.~\cite{awan2006distributed}, Gkantsidis et al.~\cite{gkantsidis2004random}, propose distributed algorithm for sampling in unstructured p2p networks, but do not consider adversarial node behavior. 

We do not attempt to present a full-fledged p2p network design for the applications outlined above (broadcast, data availability sampling, sharding, server discovery etc.) in this paper. 
Rather we posit that a uniform node sampling capability will have a central role to play in the eventual overall network design for these applications. 

We evaluate Honeybee through a custom simulator we have built.\footnote{Code will be made open source.} Compared to the baseline algorithms GossipSub and Kademlia, Honeybee achieves the same level of near-uniform sampling when all nodes in the network adhere to the protocols. However, when the network contains adversarial nodes, Honeybee outperforms GossipSub and Kademlia by 4-63\% in terms of sampling adversarial peer ratio. We define that an algorithm achieves \(\epsilon\)-uniform sampling when the sampling adversarial peer ratio from the algorithm is bounded from above by the sum of the true adversarial nodes ratio in the network and \(\epsilon\). Under such standard, Honeybee consistently achieves 0.03-uniform sampling with 5-70\% adversarial nodes in the network.

\section{System Model}\label{systemModel} 

\subsection{Network and  Security Model} \label{net_model}

We consider a network comprising of $n$ nodes, out of which a fraction $f$ of the nodes are adversarial (dishonest). 
Nodes that are not adversarial are called honest. 
We denote the set of honest and adversarial nodes by $V_h$ and $V_a$ respectively. 
$V$ denotes the set of all nodes.
Each node has a unique network address (i.e., IP address, port). A node can connect and communicate with another node if it knows the latter’s network address. Each node has a small memory of size $k$, for storing information about $k$ other nodes in the network.\footnote{We refer to the memory space also as an address table.} Apart from the network address, we assume a node also has a public, private key pair used for signing messages and issuing commitments. Since the memory space is small, an honest node cannot know the network addresses of the entire network. However, adversarial nodes can pool together the addresses they know of to conduct attacks.  

We consider Byzantine adversaries in that they can arbitrarily deviate from our proposed protocol. Examples of Byzantine behavior include not responding to or arbitrarily delaying requests or sending malicious messages to victims. We assume the network connection between nodes is reliable  (i.e., a synchronous model), and do not model message loss or delays.     

When a node first joins the network, it contacts a bootstrapping server from which it receives information (network address, public key) about $k$ random nodes in the network. It is common for practical p2p networks (e.g., Ethereum’s network) to use bootstrapping servers with hardcoded IP address to help new peers join the network~\cite{eth2netlayer}. Once a node joins the network, it must run its own discovery protocol and cannot query the bootstrapping server for fresh addresses. Note that the initial set of addresses the bootstrapping server provides can include adversarial nodes as well. 

We assume each IP address has an associated unique public key. 
The public key to IP address binding of a node is attested by the bootstrapping server through a certificate signed by the server. While it is possible to use a decentralized public-key infrastructure for this purpose~\cite{shi2022blockchain}, we consider the bootstrapping servers as the trusted certificate authority in our setting for simplicity. 

Time is divided into rounds $t=0,1,2,\ldots$. During a round, a node can send or receive messages of total size $l$, where $l$ is again a small constant that models the bandwidth constraints of the node. Lastly, we assume nodes have access to a fresh public random number each round. In practice, a new block is produced in Ethereum every 12 seconds. Thinking of 12 seconds as a round, a random number can be derived each round from the header of the block for that round. Note that we require nodes to only download the block header to compute the randomness, and not the entire block, which is particularly useful for light nodes. 
Alternative public and trusted sources of randomness may also be used as we will discuss in \S\ref{practical_consider}.


\subsection{Problem Statement} 
\label{s:probstmt}

For any node $v \in V$, let $M_v(t)$ denote the contents of the address table (i.e., memory) of node $v$ at time $t$. For any $i \in \{1, 2, \ldots, k\}$, let $M_v^i(t)$ be the $i$-th address in $M_v(t)$. 
Our objective is to design a decentralized algorithm  by which an honest node can continuously sample nodes uniformly at random from the network. Specifically, for any node $v \in V_h$, time $t$ and index $ i \in \{1, 2, \ldots, k\}$, we want 
\begin{align}
P(M_v^i(t) \in V_h) &\geq 1-f\\
P(M_v^i(t) = u) &= P(M_v^i(t) = u’)
\end{align}
for any honest nodes $u, u’ \in V_h$. The reason we lower bound the probability is achieving perfect uniform sampling would be impossible if the adversarial nodes do not participate in the protocol. On the other hand if adversarial nodes all behave honestly perfectly uniform random sampling must be possible. 

The requirements outlined above can be trivially satisfied if a node downloads random addresses from the bootstrapping servers once initially and does not update the addresses afterward. Therefore, we qualify our objective by additionally requiring that a node must add at least one fresh sample to its address table every $\Delta > 0$ rounds, i.e., 
\begin{align}
M_v(t+\Delta) \backslash M_v(t) \neq \{ \}, 
\end{align}  
for all $t > 0$. In addition, since it is infeasible for the bootstrapping servers to handle the overhead as network size increases, we preclude the trivial solution of refreshing addresses by periodically downloading random addresses from the bootstrapping servers. Furthermore, for any time $t$, let $R_v(t)$ be the most recently added address to $v$’s address table. We want the newly sampled node to be independent of the past samples, i.e., 
\begin{align}
P(R_v(t) = u | M_v(0), M_v(1), \ldots, M_v(t - \Delta)) = P(R_v(t) = u),  
\end{align} 
for all $u \in V_h$ and $t > \Delta$.

\section{\OA} \label{OA}
\subsection{Algorithm Overview}
We propose \OA. \OA~is a fully decentralized p2p algorithm that tackles the node sampling problem by conducting secure near-uniform sampling of peers from the entire network.\footnote{Note that near-uniform is defined for peer sampling in \S\ref{sec:intro} and \S\ref{s:probstmt}.}
In \OA, a node samples an address by carrying out a random walk of sufficient length on the network.
To ensure the samples are uniformly random, the graph on which the random walk is conducted is designed to be regular resulting in a doubly-stochastic probability transition matrix with a uniform stationary distribution for the Markov chain~\cite{aldous1983random,lovasz1993random}.
The graph on which the random walks are executed is different from the network's communication overlay, which can have an arbitrary topology per the application's specifications. 
To ensure that the regularity of the random walk graph is maintained and random walks are carried out correctly, \OA~relies on two key ideas: (1) random walk paths (and their lengths) are decided pseudo-randomly using verifiable random functions and public randomness source (details in \S\ref{vrw}); (2) addresses sampled by a node are independently and secretly cross-verified by other nodes. If an inconsistency is discovered, a fraud proof is issued to the network (details in \S\ref{tcc}). 

\subsection{Address Table Design}
 
A \OA~node $v$ maintains three tables: an address table, an encounter table, and a connection table.
The address table holds the addresses of the nodes most recently sampled by $v$. 
Entries in the address table are regularly replaced by fresh samples through \OA's sampling algorithm.
The encounter table holds addresses of nodes most recently visited by $v$ as part of \OA's random walks. 
Addresses in the encounter table help to bootstrap random walks when the address table does not contain sufficient addresses. 
A node's connection table contains peers with whom the node communicates.
The connection table can be a subset of addresses from the current or past address table. 

When a node joins the \OA~network, it first contacts a bootstrapping server to receive an initial set of random peer addresses for its address table. 
Bootstrapping servers are a small group of trustworthy nodes that assist other nodes with their initial configurations. 
We assume the bootstrapping server has the capability to deliver uniform peer samples to requesting nodes. 
Such a capability can be achieved by letting the bootstrapping nodes also perform periodic peer sampling, or by other means. 

The address table for node $v$ consists of two sub-tables: an outgoing address table $\Gamma^v_{out}$ that contains data for at most \(n_{out} = k/2\) peers and an incoming address table $\Gamma^v_{in}$ that contains data for at most \(n_{in} = k/2\) peers.
The outgoing address table contains node addresses that $v$ has sampled, while the incoming address table contains the addresses of nodes which have sampled $v$.  In \OA, for two nodes $u, u'$, if $u' \in \Gamma^u_{out}$ then $u \in \Gamma^{u'}_{in}$ and vice-versa.  
To ensure this property, if a node \(u\) samples a node $u'$, it first sends a request to $u'$. 
If node $u'$ accepts the request, \(u\) adds $u'$ to the outgoing address table of \(u\), while \(u'\) adds \(u\) to the incoming address table of \(u'\).
This process involves $u$ and $u'$ signing a `peering agreement' with an expiration time (e.g., $k\Delta$). 
If the peering agreement is not renewed before its expiry, then $u$'s address in the address table of $u'$ (and vice-versa) is considered invalid. 
Address table updates in \OA~are thus bilateral. 
Conversely, when $u$ updates its address table to replace $u'$ by another address, $u$ communicates the update to $u'$ so that $u'$ can remove $u$ from its address table. 
Addresses that are not removed cooperatively from a node's address table naturally time out and become verifiably invalid.

If \(u\) has \(u'\) in its address table, \(u\) stores the following data for $u'$ in the address table: public key of \(u'\) (certified by a bootstrapping server, or by a decentralized certificate authority), a snapshot of the address table of \(u'\), IP and port address of \(u'\), and a peering agreement signed by both \(u\) and \(u'\) with an expiry timestamp. 
A node's address table snapshot contains everything in its address table except for its peers' address table snapshots. 
These snapshots help to verify the correctness of random walks as discussed next.\footnote{We refer to the shared copy of an address table as a snapshot because, while it may be susceptible to modification or tampering, it can also serve as evidence against the table's owner later, as it is digitally signed.} 

\begin{algorithm}[!tbp]
\DontPrintSemicolon
\SetKwInOut{Input}{input}\SetKwInOut{Output}{output}
\Input{peers $\Gamma^v_\text{curr}$ in outgoing address table of current epoch; 
shared pseudo-randomness seed $\mathcal{R}$;
$v$'s secret key $SK_v$}
\Output{updated set of peers $\Gamma^v_\text{next}$ for next epoch}
\If {$\mathcal{R}$ indicates $v$ is a currently eligible node} {
\tcc{$v$ infers the random walk path length $p$ from $\mathcal{R}, v, SK_v$.}
$p_c \leftarrow 0 $ \tcc{Hop counter of the random walk}
$u^* \leftarrow v $ \tcc{$u^*$ stores node on current hop of the random walk} 
\tcc{$d$ stores the index of node on first hop of the random walk} 
\While {$p_c < p$} {
\tcc{$v$ performs one hop of the iterative verifiable random walk and updates $p_c, u^*$; $v$ issues a fraud proof against $u^*$ when snapshots from VRW or TCC for the same node $u^*$ differs materially (more details in Appendix \ref{alg1_details})}
}
$\Gamma^v_\text{next} \leftarrow \Gamma^v_\text{curr} \backslash \{ \Gamma^v_\text{curr}(d) \} \cup \{ u^* \}$\; 
}
\caption{Algorithm outline for updating entries of the outgoing address table of node $v$ in each epoch.}
\label{alg1}
\end{algorithm}

In Algorithm \ref{alg1}, we present the overall template for outgoing address table update of \OA. 

\subsection{Verifiable Random Walks (VRW)} \label{vrw}

In each round of \OA, a fraction $\eta$ of nodes in the network attempt to obtain a fresh sample for their outgoing address table. 
Restricting address table updates to only a fraction of the nodes avoids excessive churn in the address tables which in turn helps maintain the stationary distribution of the random walks to be uniform.  
A node determines its eligibility to conduct a random walk during a round by computing a verifiable random function (VRF) over the shared public randomness for the round.   
When a node $u$ becomes eligible during a round, it performs a random walk to replace one of its outgoing table peers. 

A key challenge in performing a random walk is adversarial routing: 
if a random walk from an honest node meets a dishonest node, the dishonest node can route the random walk to another dishonest node and eventually guide the random walk to a dishonest destination. 
To avoid adversarial routing, in \OA~the length $l$ and path $p$ of a node $u$'s random walk are determined pseudo-randomly using the round's public randomness and the VRF.
The input to the VRF includes the current round number and public randomness value. 
The path $p$ of the random walk is implicitly specified as a sequence $x_1, x_2, \ldots, x_l$ of indices where each $x_i \in \{1, 2, \ldots, k\}$ for all $i$.
To conduct a random walk, the node $u$ visits each node of $p$ sequentially. 
Supposing node $u$ visits node $u'$ for hop $i-1$, and moves to node $u''$ for hop $i$. 
Recall that $u'$ has a snapshot of the address table of $u''$, and vice-versa. 
When $u$ contacts $u''$, node $u''$ first verifies that the most recent hop was executed correctly, i.e., the $x_{i-1}$-th entry in the address table of $u'$ is $u''$. 
Node $u''$ also verifies whether the current hop number for $u$ is not greater than $l$. 
To do this, $u$ submits snapshots of the address tables from each of the $i-1$ nodes visited so far in the walk to $u''$. 
The address table snapshots collected by $u$ are signed by their respective owners. 
A table consistency check algorithm (\S\ref{tcc}) verifies that the signed address table snapshots are valid, and thus verifies the random walk length.  
The walker $u$ then requests node $u''$ to provide the IP address of the $x_i$-th address within the address table of $u''$. 
To verify whether the received address is correct, $u$ checks the address received from $u''$ against the snapshot of the address table of $u''$ received from $u'$. 
If the address is verified to be correct, node $u$ then proceeds to the $i+1$-th step of the walk. 
Finally, when $u$ reaches the last node of the walk it requests the node to be added as a new sample within $u$'s table. 
If the requested node accepts the request, the initiator node $u$ and the destination node add each other to their address tables and establish a peering agreement.
A snapshot of $u$'s updated address table is shared with all peers in its address table. 
For a node, we call the time between successive roles as an eligible node as an epoch. 

\begin{figure}[!tbp]
  \centering
  \subfloat[Honeybee Random Walk]{\includegraphics[width=0.36\textwidth]{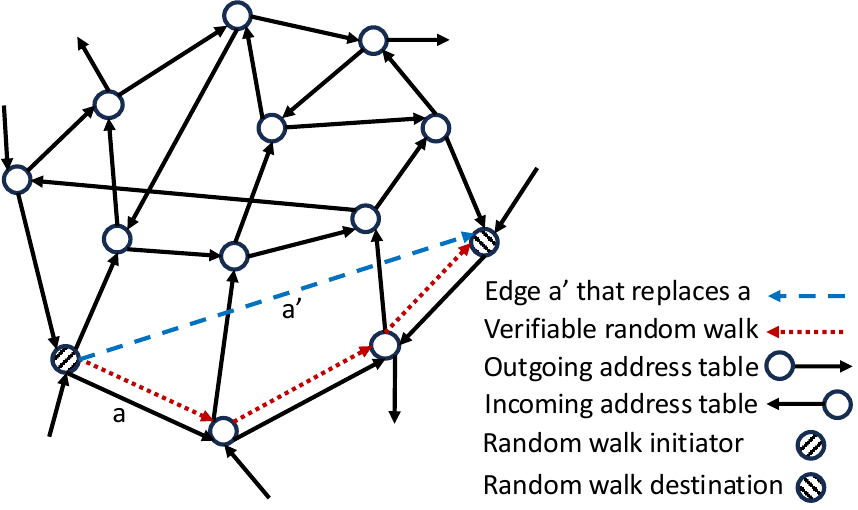}}
  \hfill
  \subfloat[Honeybee VRW]{\includegraphics[width=0.19\textwidth]{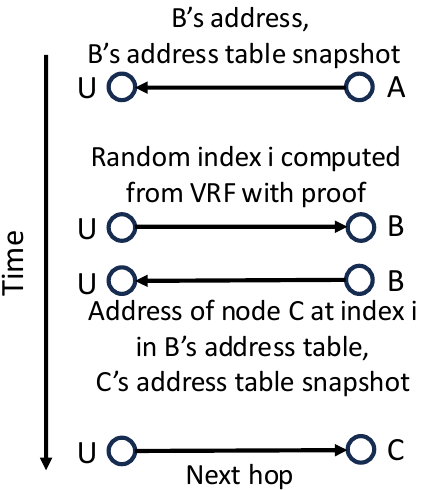}}
  \hfill
  \subfloat[Honeybee TCC]{\includegraphics[width=0.27\textwidth]{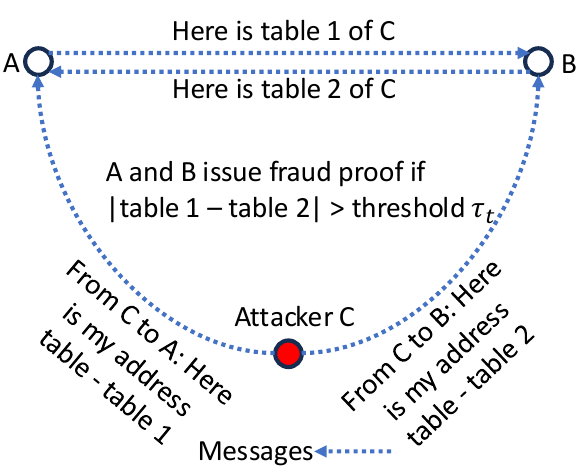}}
  \caption{Honeybee illustrations: (a) an example of \OA~verifiable random walk (VRW) sampling with three hops; (b) how VRW works from hop to hop; (c) how table consistency check (TCC) works among nodes.}
  \label{fig:diagrams}
\end{figure}

 In Fig. \ref{fig:diagrams}a, we give an example of \OA~verifiable random walk sampling with three hops, and the destination node replaces the node on the first hop in the outgoing address table of the initiator node at the end of the sampling. Fig. \ref{fig:diagrams}b shows how a verifiable random walk works.

\subsection{Table Consistency Checks (TCC)} \label{tcc}

A dishonest node can store multiple address tables, and use different address tables to handle different requests from different nodes for various purposes (e.g., traffic attraction, adversarial routing, etc.). 
We refer to the situation where a node uses more than one address table as ``equivocal tables.'' To prevent equivocal tables, we employ table consistency checks (TCC) in \OA. 
When a node $u$ visits a node $u'$ during a random walk, $u$ and $u'$ determine whether there are any addresses common to the address and encounter tables of $u$ and $u'$. 
If an overlapping node \(u''\) is found in the tables of both \(u\) and \(u'\), subsequently \(u\) and \(u'\) will further compare the address table snapshots of node \(u''\) as stored by $u$ and $u'$ to check if there is any difference. 
If the difference is higher than a threshold \(\tau_t\) (i.e., more than $\tau_t$ addresses are different in the snapshot versions of $u$ and $u'$), then \(u\) and \(u'\) will issue a fraud proof against node \(u''\) and the system may choose to slash node \(u''\). 
\(\tau_t\) is predefined by the system, and \(\tau_t\) increases as \(t\) (time difference between two snapshot versions) increases. \(t\) is bounded from above by a system-defined parameter \(t_{max}\) such that nodes only need to store a limited history of their past random walks.\footnote{Raising \(t_{max}\) improves system security but increases storage demands on nodes. In the analysis (Appendix~\ref{theory}), we use \(t_{max}=\frac{k}{2\eta}\).}
If the ratio of different peers in two address table snapshots of node \(u''\) within time \(t\) exceeds \(\tau_t\), then we consider the evidence is significant enough for nodes to issue a fraud proof for node \(u''\). Node \(u''\) can refute the TCC fraud proof against it only if it can provide legitimate random walk history that can explain the difference. Fig. \ref{fig:diagrams}c gives an example of a table consistency check. 

\section{Evaluation} \label{eva}
In this section, we evaluate \OA~and the baseline algorithms and compare their performances. In \S\ref{setup}, we introduce the baseline algorithms and present the experimental setup. We then evaluate \OA~and the baseline algorithms from different perspectives and provide the results in \S\ref{res}.

To understand the effectiveness of \OA~theoretically, we conduct an analysis of \OA~and include the results in Appendix~\ref{theory}.

\subsection{Experimental Setup} \label{setup}
Here, we describe our experimental setup, which includes the baseline algorithms, network layouts, adversary configurations , and adversary strategies.

\subsubsection{Baselines} \label{baselines}
Since the main goal of \OA~is to conduct near-uniform peer sampling in p2p networks, we compare \OA~with two arguably most potential candidate algorithms under such setting \cite{Neu_2022} - Kademlia \cite{maymounkov2002kademlia} and GossipSub \cite{vyzovitis2020gossipsub}. We briefly address them below.
\begin{itemize}
    \item Kademlia: Kademlia is one of the most popular p2p protocol in today's Internet. It is used in various systems including Ethereum, Swarm, Storj, IPFS, etc. A Kademlia node has a binary node ID that is randomly assigned when the node joins the network. To route messages, each Kademlia node has a routing table consists of sub-tables called ``k-buckets,'' and the number of k-buckets it has is equal to the length of its node ID. The \(i\)-th k-bucket of node A stores peers with node IDs that share the first \(i-1\) bits with node A. A Kademlia node discovers new peers mainly with lookups. When node \(A\) performs a lookup on node ID \(B\), \(A\) sends the lookup message to the neighbor(s) whose node ID is closest to \(B\) in terms of XOR distance. The neighbor(s) and the nodes on the subsequent hops repeat this procedure until they find the closest node(s) to \(B\). For more details about Kademlia, we refer the reader to the Kademlia paper \cite{maymounkov2002kademlia}.

    In our evaluation, the Kademlia nodes each has 14 k-buckets and each k-bucket has size of 3 (note that some buckets may never have enough peers to reach the size limit) such that the baseline algorithms have similar memory space sizes, and we use \(\alpha=3\) for all lookups to make Kademlia more efficient (i.e., we send out three parallel queries to perform efficient lookups).

    \item GossipSub: GossipSub is arguably the most renowned publish-subscribe gossip network in today's Internet. It is used in libp2p, an open source project from IPFS. A GossipSub node has its mesh connections and gossip connections. The mesh connections are bidirectional connections, and nodes use them to send full messages. The gossip connections are unidirectional, and nodes use them to send metadata only. A GossipSub node discovers new peers mainly with peer exchanges, in which a node shares the information of some of the peers it knows with others. For more details about GossipSub, we refer the reader to the GossipSub paper \cite{vyzovitis2020gossipsub}.
    
    In our evaluation, we simulate peer discovery in GossipSub with peer exchanges. We simplify the scoring function by giving every peer the same score to fit GossipSub into the peer sampling background. The scoring function from GossipSub may not help honest nodes in the peer sampling setting. The scoring function uses parameters such as time in mesh, first message deliveries, and mesh message delivery rates. With these parameters, dishonest nodes can strategically exploit the scoring function and behave very well in terms of scores before eclipsing an honest node. We use \(D\_high=12\), \(D=8\), and \(D\_low=6\) for GossipSub nodes, which are the same as the values used in the GossipSub paper.\footnote{\(D\_high\), \(D\), and \(D\_low\) refer to the admissible mesh degree bounds~\cite{vyzovitis2020gossipsub}.}
\end{itemize}

\subsubsection{Network} \label{network_layout}
We built a discrete-event network simulator using Python based on the model description in \S\ref{net_model}. We simulate three p2p networks - the Kademlia network, the GossipSub network, and the \OA~network. Each of these networks consists of 16,384 (\(=2^{14}\)) nodes, and 17 of them (\(\approx0.1\%\)) are the bootstrap nodes. The Kademlia network simulates the Kademlia network as discussed in \S\ref{baselines}, and the GossipSub network simulates one overlay of the GossipSub network as discussed in \S\ref{baselines}. The \OA~network simulates \OA~as discussed in \S\ref{OA}. Similar to the baselines, the \OA~nodes each uses an address table size of 24. The address table of a \OA~node consists of the outgoing address table and the incoming address table, each contains at most 12 nodes. In all three networks, a node's routing table is defined as the peers whom it knows and has access to (i.e., stores IP addresses), and we assume there is no churn from the joining and leaving of nodes. 

\subsubsection{Adversary} \label{attack_config}
For each of the three p2p networks, we simulate scenarios in which attackers control 5\%, 10\%, 20\%, 30\%, 40\%, 50\%, 60\%, 70\%, and 80\% of the total nodes. The attacker-controlled dishonest nodes (e.g., Sybil nodes) are randomly chosen from all the nodes at the beginning of the simulation.\ To disturb the network and compromise potential victim node(s), the dishonest nodes employ strategies addressed in \S\ref{attack_strategy} during the simulation. Attackers have full control of their dishonest nodes and are able to build a desired topology with their dishonest nodes. We consider two types of initial dishonest node layouts: mixed layout and clustered layout.
\begin{itemize}
    \item Mixed layout: the dishonest nodes are mixed into the p2p network. All the nodes in the network are connected to each other in a random way, and the dishonest nodes are simply among them. We consider this layout the most natural initial attacker layout. 
    \item Clustered layout: the dishonest nodes form one or more clusters. Most (\(\approx\)98\%) of the dishonest nodes only connect to each other while 2\% of the dishonest nodes have connections to the honest cluster. The dishonest nodes whose connections include a connection to the honest nodes are named gateway Sybil nodes. The dishonest nodes who do not have a connection with the honest nodes are named trap Sybil nodes. The two types of dishonest nodes may switch their roles during the simulation.
\end{itemize}

\begin{figure}[!tbp]
  \centering
  {\includegraphics[width=0.38\textwidth]{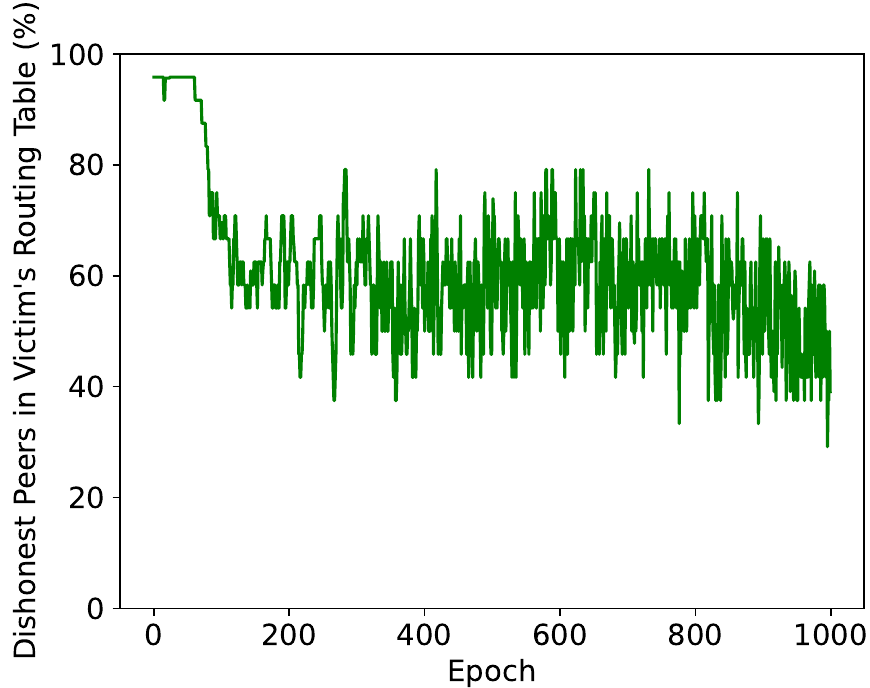}}
  \caption{In a Honeybee network with 50\% dishonest nodes, all dishonest nodes cluster to attack a victim whose initial address table contains only one honest outgoing neighbor (an unlikely adverse scenario). We conduct five separate such experiments and present the worst-case outcome for the victim.}
  \label{fig:cluster_attack}
\end{figure}

At first glance, a lethal attack in the clustered layout setting in which an honest victim node whose random walks venture into the dishonest nodes' cluster is completely trapped (eclipsed) within the dishonest node cluster appears feasible.  
However, in Fig. \ref{fig:cluster_attack}, we show that such an attack is unlikely to succeed if the number of dishonest nodes is not significantly greater than the number of honest nodes (the settings considered in our experiments). 
If the victim node has even a single outgoing neighbor in the honest cluster (i.e., is not totally eclipsed), incoming connections requests from the honest cluster quickly increase the fraction of honest peers within the victim's address table. 
Moreover, with verifiable random walks, attracting a target victim node's walks inside the adversary's cluster is also not easy. 
In our evaluations, we therefore focus on the mixed layout as 
the more damaging setting to the honest nodes. 

We focus on two types of victims: single victim node and multiple victim nodes. In the first case, dishonest nodes target a random honest node throughout the simulation (e.g., attackers attempt to compromise a client). In the second case, dishonest nodes target all honest nodes throughout the simulation (e.g., attackers attempt to compromise an organization or a company). In both cases, the goal of attackers is to eclipse the victim(s) by inserting as many dishonest nodes into the address tables of the victim(s) as possible.\footnote{In this paper, an eclipse refers to a full eclipse---a situation where all the neighbors of one or more victims are dishonest.}

\subsubsection{Adversary Strategies} \label{attack_strategy}
In our evaluation, attackers employ a wide range of adversary strategies using the nodes they control. We categorize the adversary strategies into two types: active strategies and passive strategies. 
The active adversary strategies include request flood, adversarial routing, adversarial peer selection, and equivocal table defined as follows: 
\begin{itemize}
  \item Request flood: attackers can repeatedly send requests (e.g., connection requests) to the victim node(s) from their pseudonymous identities until the victim node(s) get eclipsed or the attackers achieve their goal through other means. It is difficult for the honest nodes to detect and defend against this strategy since, in a permissionless p2p network, attackers may own a large number of pseudonymous identities and can send requests from different identities. It is challenging to distinguish the dishonest nodes from the honest nodes before attackers cause actual damage to the victim node(s).
  \item Adversarial routing: once a request reaches a dishonest node, the dishonest node may route the request in its favor to achieve certain goals. The goals include but are not limited to: preventing the request from reaching its intended destination, guiding the request to a dishonest destination, causing overhead (e.g., delays) for the request initiator, and causing inaccurate judgment (e.g., inaccurate neighbor scoring) for the request initiator. 
  \item Adversarial peer selection: an honest node should select its peers in a random way or based on certain bona fide rules. An honest node should not make its peer selection decisions based on the identity of the candidate peers (and it should not be able to). However, a dishonest node may make such decisions based on the identity of the candidate peers. For example, a group of dishonest nodes may choose to add each other to their routing tables to form a dishonest cluster. 
  \item Equivocal table: a dishonest node can keep multiple copies of routing tables, each storing different peers, and use different routing tables to route different requests. This strategy can prevent the request from reaching its intended destination, guide the request to a dishonest destination, help other dishonest nodes with load-balancing, cause overhead (e.g., delays) for the request initiator, and cause inaccurate judgment (e.g., inaccurate neighbor scoring) for the request initiator. Equivocal table attacks are dangerous because they are important means to further attacks. 
\end{itemize}

The passive adversary strategies include selective request acceptance, adversarial recommendation, and black hole defined as follows: 
\begin{itemize}
  \item Selective request acceptance: an honest node should accept/reject a connection request in a random way or based on certain bona fide rules. An honest node should not make its acceptance/rejection decisions based on the identity of the request initiator (and it should not be able to). However, a dishonest node may make such decisions based on the identity of the request initiator. For example, attackers and their pseudonymous identities may choose to reject requests from all the honest nodes except for its targeted victim node(s). 
  \item Adversarial recommendation: upon request, an honest node should recommend (i.e., share the information of) a peer to another node in a random way or based on certain bona fide rules. An honest node should not make its recommendation decisions based on the identity of the request initiator and (and it should not be able to). However, a dishonest node may make such decisions based on the identity of the request initiator. For example, an honest victim's dishonest neighbor may choose to recommend other dishonest nodes to the victim. This strategy is similar to request flood but in a passive manner. 
  \item Black hole: once a victim request reaches a dishonest node, the dishonest node may ignore the request completely (i.e., being unresponsive to the request). This strategy can cause overhead (e.g., delays) and inaccurate judgment (e.g., inaccurate neighbor scoring) for the request initiator. In comparison with adversarial routing, this strategy is less harmful because an honest node can simply drop its random walk. Modern p2p networks use reputation or credit systems to help mitigate black hole attacks ~\cite{9443953, dixit2015review}.
\end{itemize}

\subsection{Results} \label{res}
\begin{figure}[!tbp]
  \centering
  \subfloat[Honeybee]{\includegraphics[width=0.33333\textwidth]{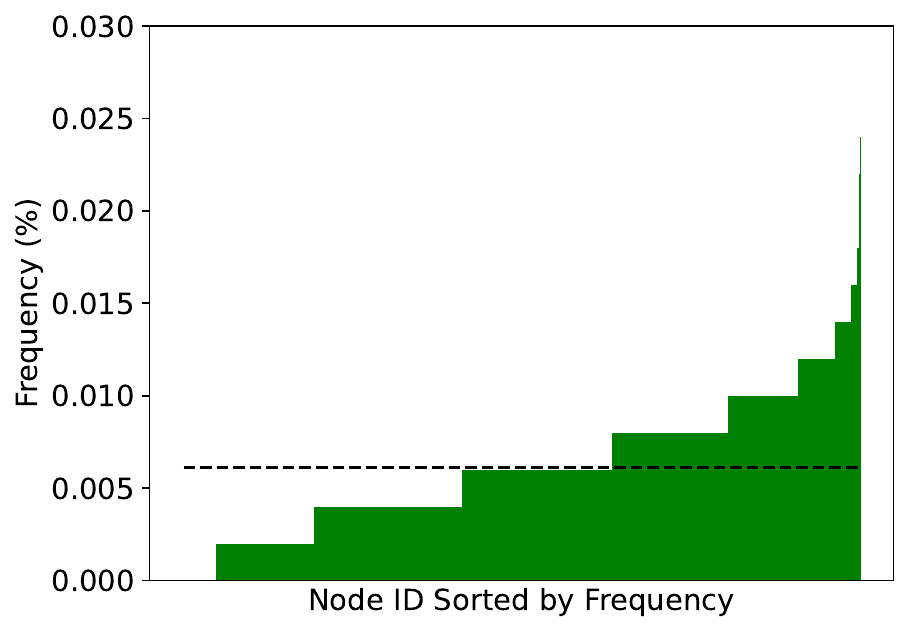}}
  \hfill
  \subfloat[GossipSub]{\includegraphics[width=0.33333\textwidth]{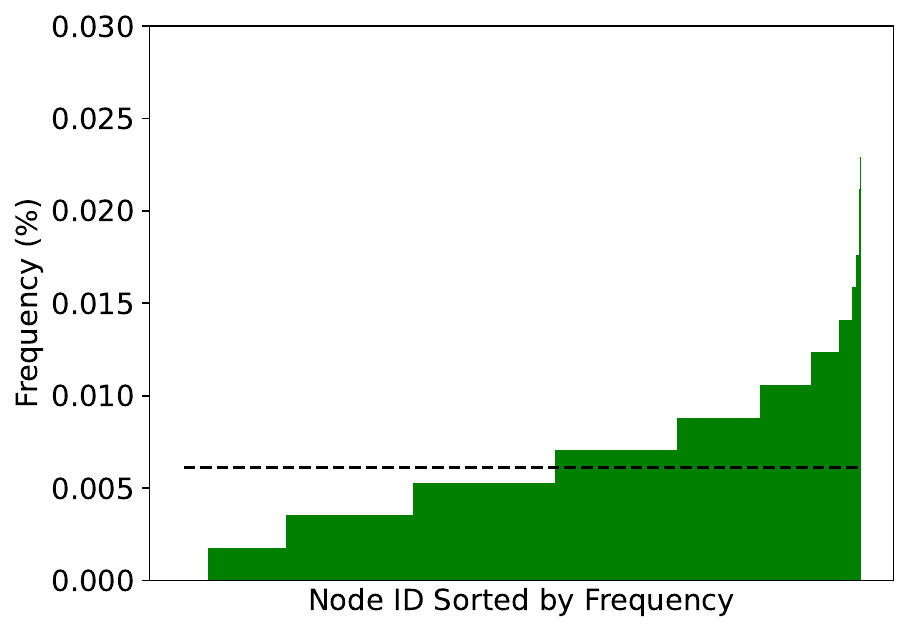}}
  \subfloat[Kademlia]{\includegraphics[width=0.33333\textwidth]{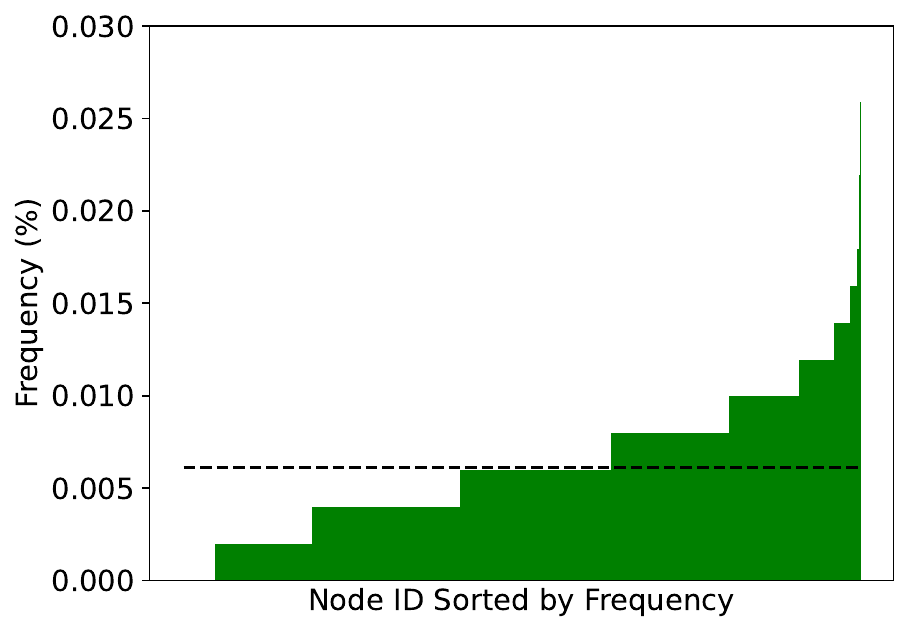}}
  \caption{Nodes adhere to protocol: Sampling distribution from a random observation node in 100,000 epochs. Every node follows its protocol. Node IDs are sorted by frequency in ascending order. True uniform sampling distribution is shown as a dashed line.}
  \label{fig:sample_dist_honest}
\end{figure}
\subsubsection{Sampling Distributions} Fig. \ref{fig:sample_dist_honest}a, \ref{fig:sample_dist_honest}b, and \ref{fig:sample_dist_honest}c show the sampling distributions for an arbitrarily chosen observation node with \OA, GossipSub, and Kademlia in a network of 16,384 honest nodes for 100,000 epochs. In each setting, all nodes start (i.e., epoch 0) with random routing table configurations, and all nodes behave according to their protocol. In the figures, the true uniform sampling distribution (i.e., all nodes except for the observation node are sampled with equal probability) is shown as a dashed line. We observe that \OA~and the two baseline algorithms have similar sampling distributions when all nodes adhere to their protocol. 

\begin{figure}[!tbp]
  \centering
  \subfloat[Total variation distances]{\includegraphics[width=0.3333\textwidth]{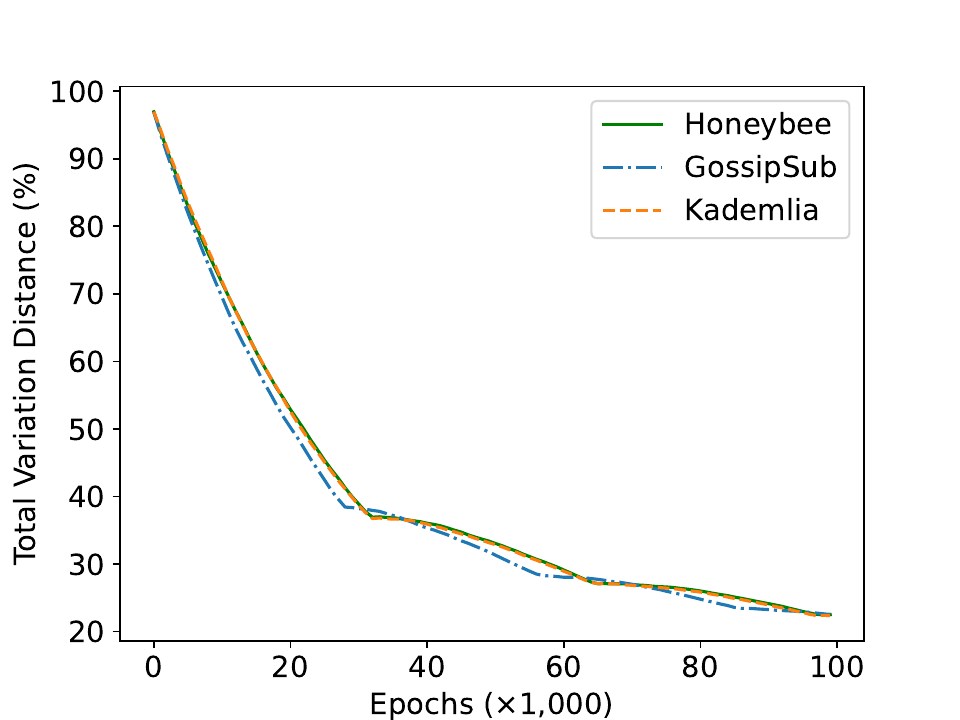}}
  \hfill
  \subfloat[\(\chi^2\) test results]{\includegraphics[width=0.3333\textwidth]{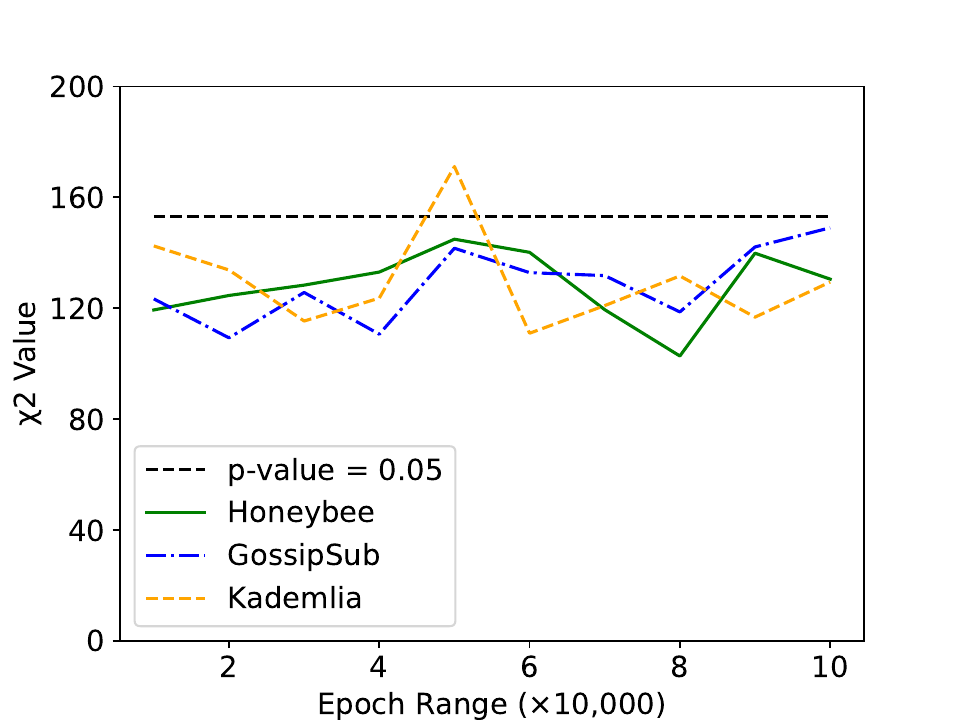}}
  \subfloat[\(\chi^2\) with 30\% idle nodes]{\includegraphics[width=0.3333\textwidth]{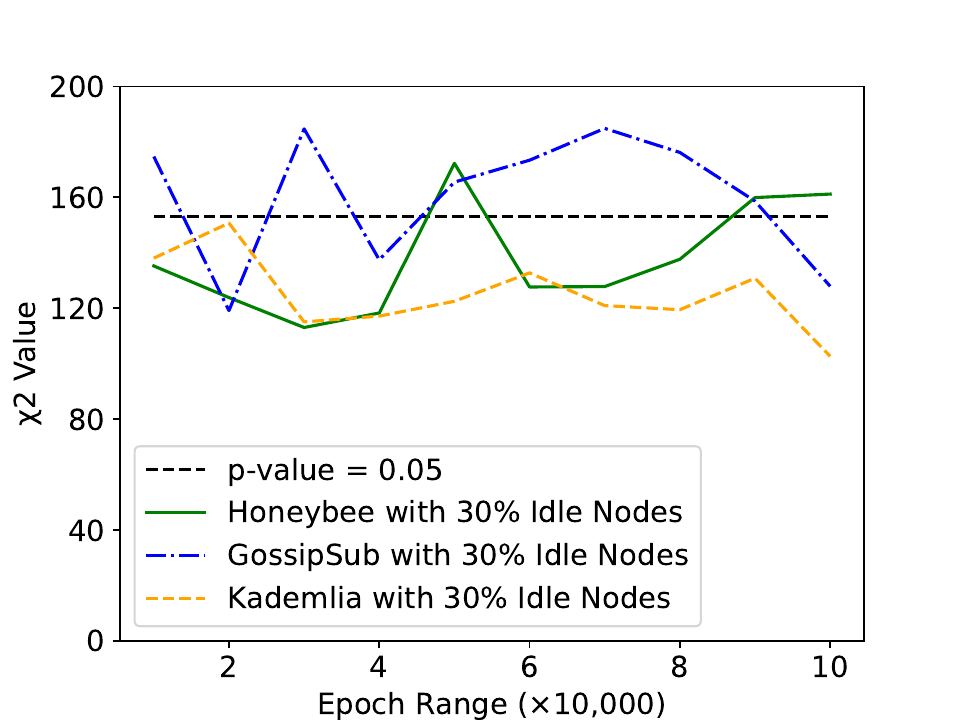}}
  \caption{Quality of sampling when nodes adhere to protocol: (a) displays comparisons of \OA, GossipSub, and Kademlia with the true uniform sampling distribution; (b) displays the Chi-Square test results with p-value of 0.05 shown as a dashed line; (c) displays the Chi-Square test results when 30\% of the nodes are idle.}
  \label{fig:sample_quality}
\end{figure}

To compare the sampling distributions of \OA~and the baseline algorithms with the true uniform sampling distribution, Fig. \ref{fig:sample_quality}a plots the total variation distances (TVDs) between the three algorithms and the true uniform sampling distribution from the above experiment. The TVDs suggest that sampling distributions from \OA~and the baseline algorithms converge to the true uniform sampling distribution throughout the 100,000 epochs in a similar pattern. At the beginning of the simulations, the TVDs from \OA~and the baseline algorithms are higher than 95\%. At 100,000 epochs, the total variance distances from \OA~and the baseline algorithms decrease to roughly 23\%. We also extended the simulations to 3 million epochs, and the TVDs from the three algorithms decrease to 3-5\%.

To examine the sampling distributions in detail, we conduct Chi-Square tests on the samples from \OA~and the baseline algorithms. We combine neighboring cells by dividing the node IDs into 127 node ID intervals of equal length (since the number of node IDs minus the observation node is divisible by 127). We divide the 100,000 epochs into 10 time intervals of equal length. For each time interval, we examine the samples from that particular time interval. In Fig. \ref{fig:sample_quality}b, we plot the Chi-Square values across the 10 time intervals for \OA~and the baseline algorithms. In Fig. \ref{fig:sample_quality}c, we plot the Chi-Square values across the 10 time intervals for \OA~and the baseline algorithms when 30\% of the nodes are idle (i.e., do not actively conduct sampling). We observe that, for \OA~and the baseline algorithms, we cannot reject the null hypothesis of uniform sampling with sufficient evidence when all nodes conduct sampling. When there are 30\% of idle nodes, the Chi-Square values increase slightly for \OA~and GossipSub. The takeaway here is that \OA~and the baseline algorithms can all conduct near-uniform peer sampling in the absence of adversarial nodes. 

\begin{figure}[!tb]
  \centering
  \subfloat[5\% dishonest nodes]{\includegraphics[width=0.333\textwidth]{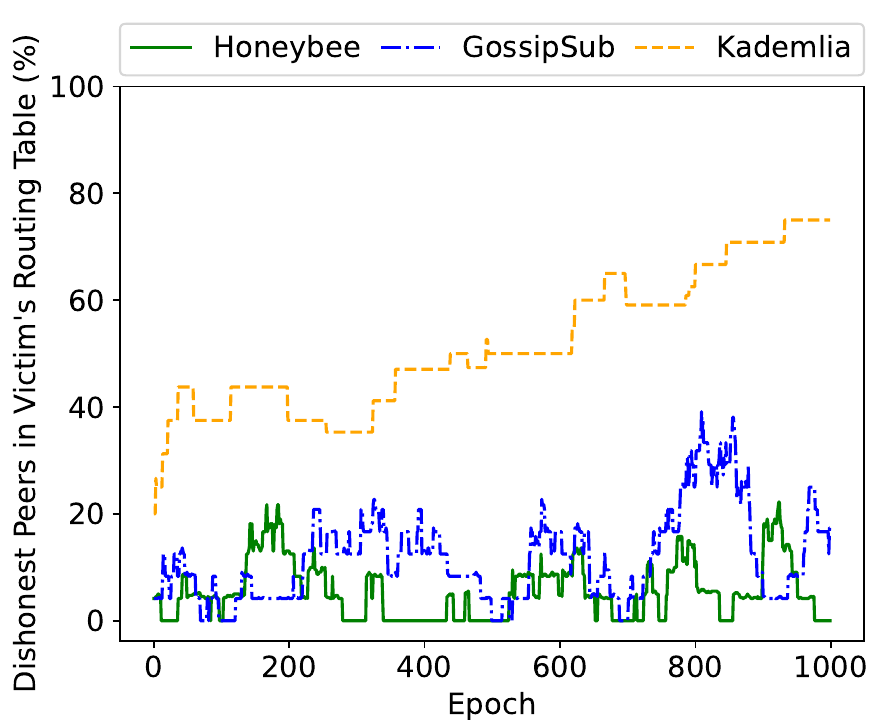}}
  \hfill
  \subfloat[10\% dishonest nodes]{\includegraphics[width=0.333\textwidth]{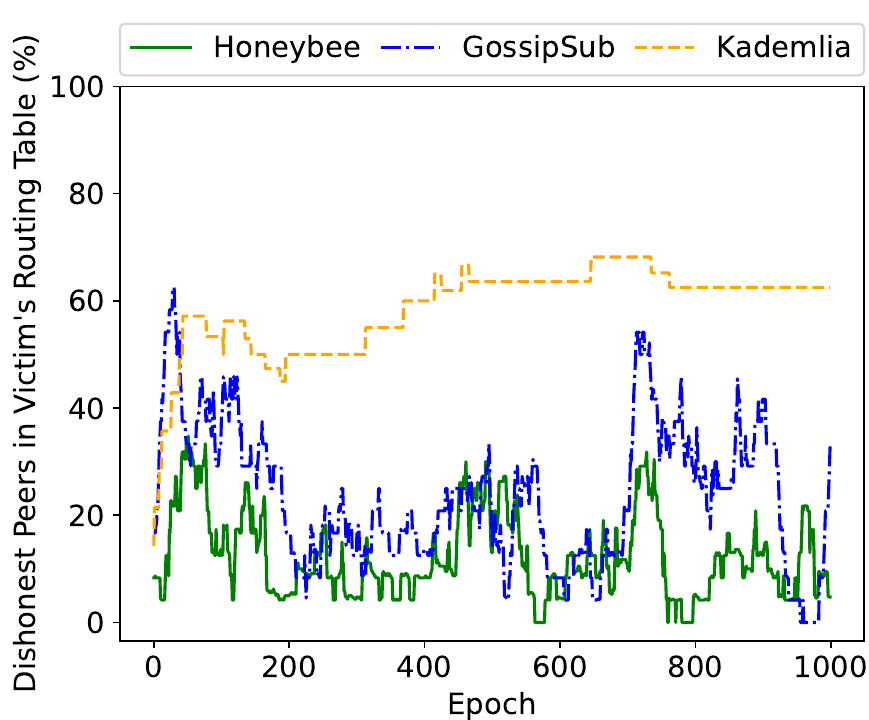}}
  \subfloat[20\% dishonest nodes]{\includegraphics[width=0.333\textwidth]{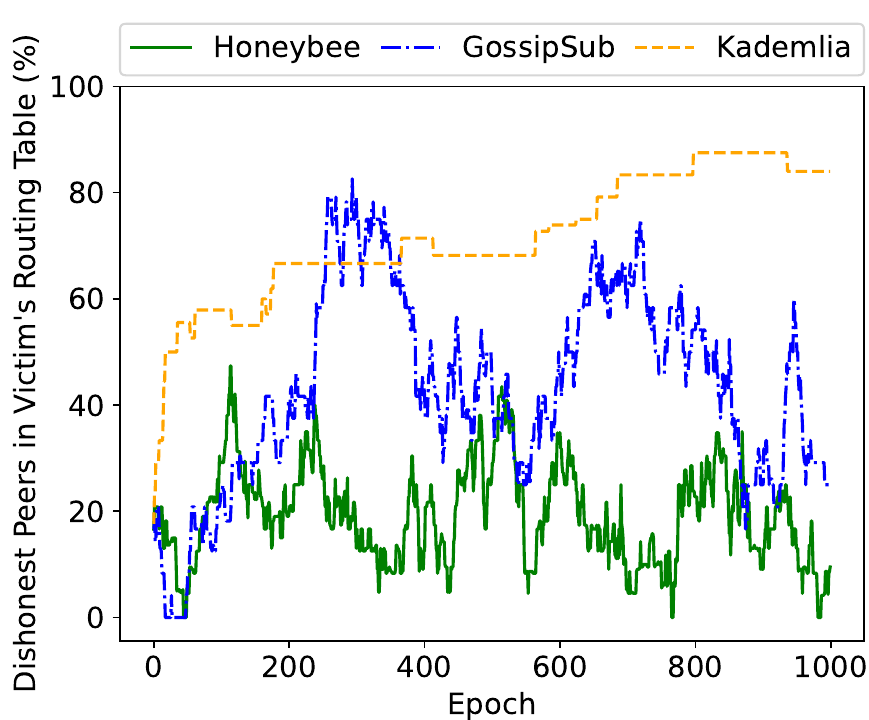}}
  \caption{An honest node attacked by 5\%, 10\%, and 20\% of dishonest nodes: Single random honest node under attack in \OA, GossipSub, and Kademlia.}
  \label{fig:one_victim_05_10_20}
\end{figure}

\begin{figure}[!tb]
  \centering
  \subfloat[30\% dishonest nodes]{\includegraphics[width=0.333\textwidth]{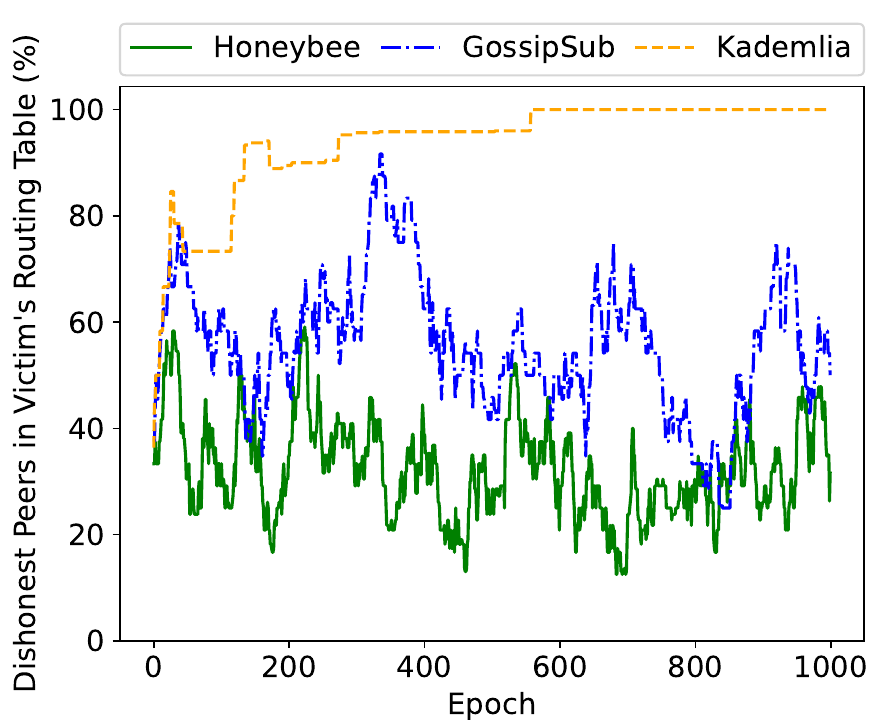}}
  \hfill
  \subfloat[40\% dishonest nodes]{\includegraphics[width=0.333\textwidth]{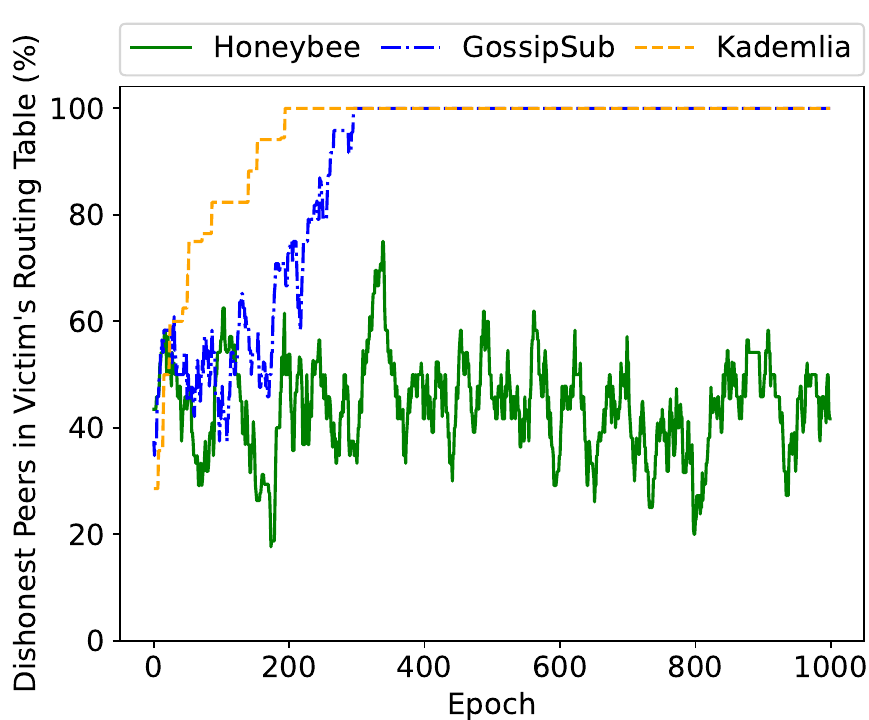}}
  \subfloat[50\% dishonest nodes]{\includegraphics[width=0.333\textwidth]{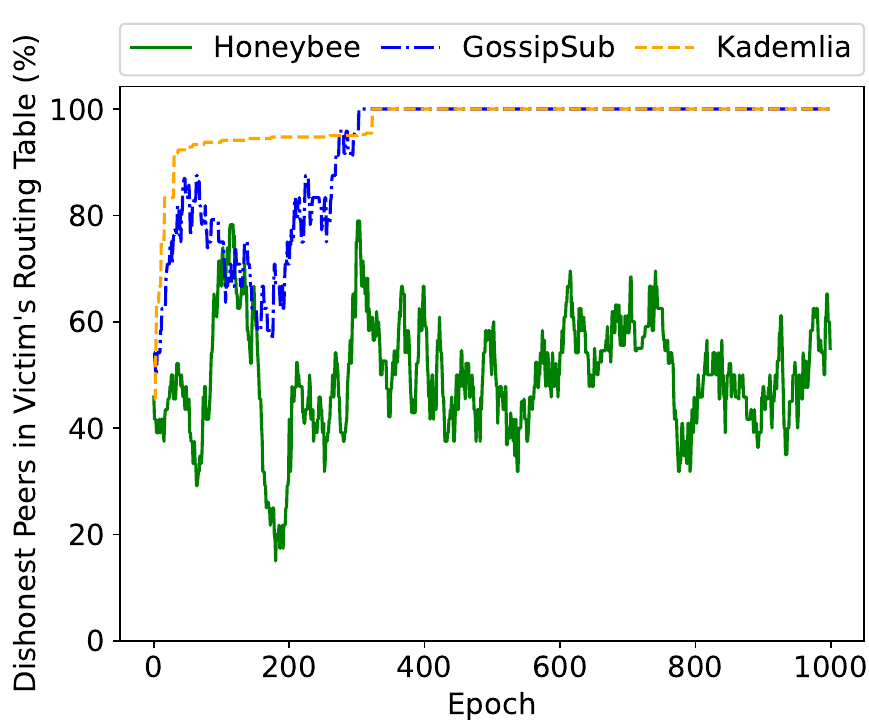}}
  \caption{An honest node attacked by 30\%, 40\%, and 50\% of dishonest nodes: Single random honest node under attack in \OA, GossipSub, and Kademlia.}
  \label{fig:one_victim_30_40_50}
\end{figure}

\begin{figure}[!tb]
  \centering
  \subfloat[60\% dishonest nodes]{\includegraphics[width=0.333\textwidth]{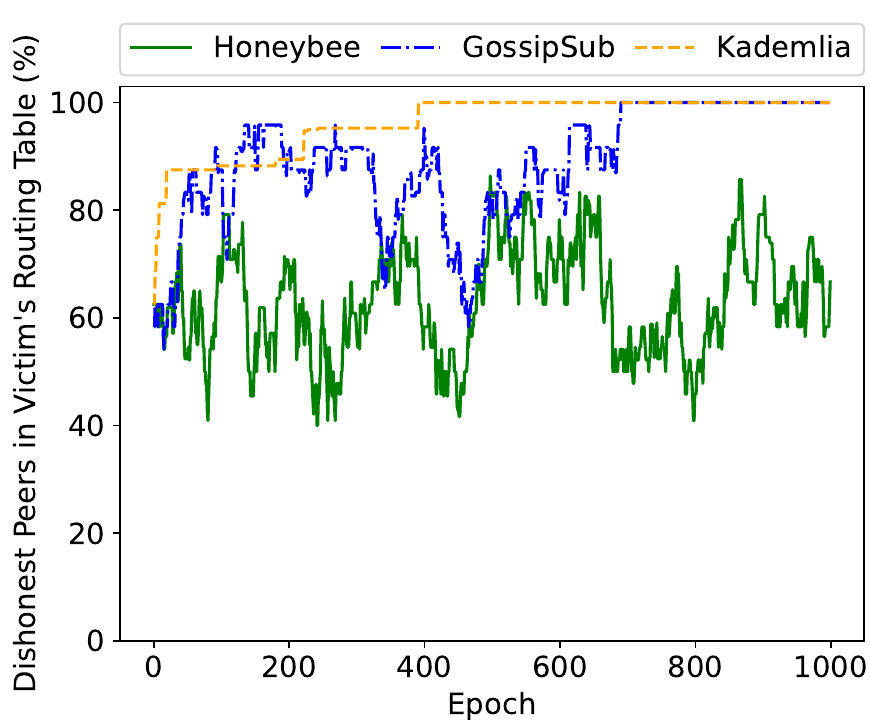}}
  \hfill
  \subfloat[70\% dishonest nodes]{\includegraphics[width=0.333\textwidth]{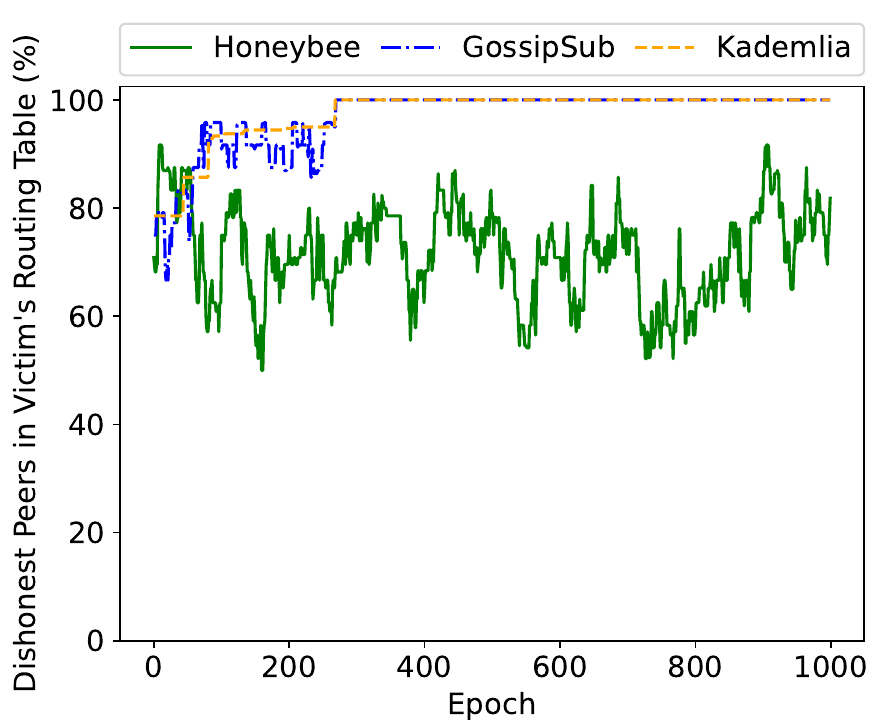}}
  \subfloat[80\% dishonest nodes]{\includegraphics[width=0.333\textwidth]{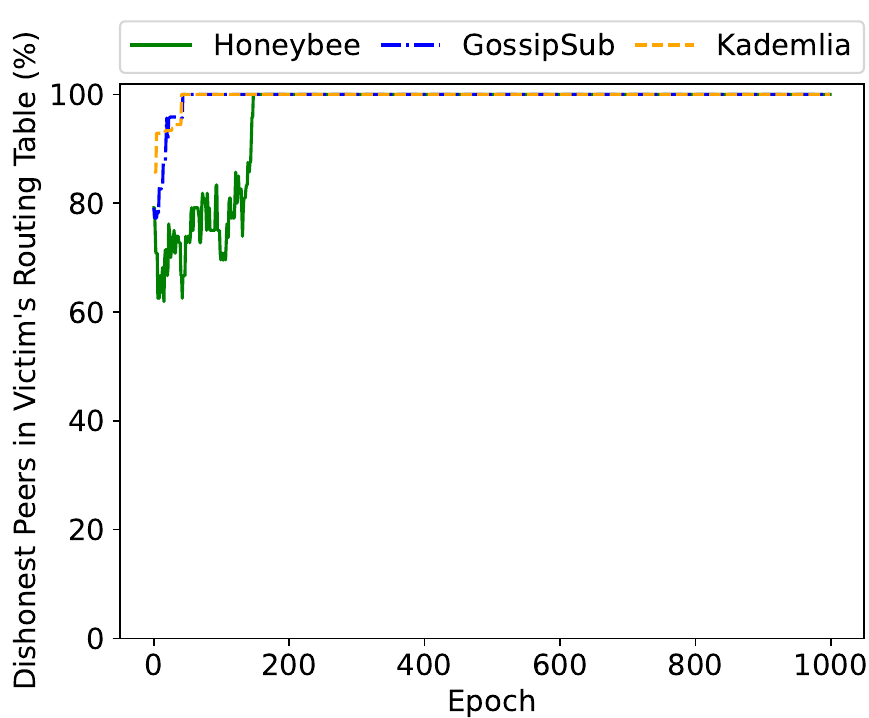}}
  \caption{An honest node attacked by 60\%, 70\%, and 80\% of dishonest nodes: Single random honest node under attack in \OA, GossipSub, and Kademlia.}
  \label{fig:one_victim_60_70_80}
\end{figure}

\subsubsection{Attack - Single Victim} Fig. \ref{fig:one_victim_05_10_20}, Fig. \ref{fig:one_victim_30_40_50}, and Fig. \ref{fig:one_victim_60_70_80} show the ratio of dishonest peers in a random honest victim node's routing table in 1,000 epochs for \OA~and the baseline algorithms in a network of 16,384 nodes, and 5\%, 10\%, 20\%, 30\%, 40\%, 50\%, 60\%, 70\%, and 80\% of the total population consists of dishonest nodes. The goal of all the dishonest nodes is to eclipse the victim node. All nodes start with random routing table configurations. We observe that, under different levels of dishonest node percentage, \OA~consistently outperforms the baseline algorithms in terms of the mean ratio of dishonest peers in the victim's routing table. 

\begin{table}[!tbp]
\begin{tabular}{llllllllll}
\hline
          & \multicolumn{9}{c}{Percentage of dishonest nodes in the network}               \\
          & 5\%     & 10\%    & 20\%    & 30\%    & 40\%    & 50\%    & 60\% & 70\% & 80\% \\ \hline
Honeybee  & 6.16\%  & 11.49\% & 21.91\% & 32.25\% & 40.90\% & 51.14\% & 61.44\% & 71.73\% & 94.62\% \\
GossipSub & 9.99\%  & 25.12\% & 42.24   & 55.52\% & 69.64\% & 81.79\% & 91.64\% & 96.26\% & 99.49\% \\
Kademlia  & 56.26\% & 63.22\% & 85.12\% & 87.15\% & 97.36\% & 95.23\% & 97.73\% & 98.80\% & 98.64\% \\ \hline
\end{tabular}
\caption{Mean ratio of dishonest peers in victim's routing table over 1,000 epochs: Five honest nodes are randomly sampled as victims in five separate simulations with \OA, GossipSub, and Kademlia. We consider different percentage of dishonest nodes in the entire network ranging from 5\% to 80\%.}
\label{tbl:one_victim_mean}
\end{table}

To demonstrate that the results apply across different nodes in different simulations, we randomly sample five honest nodes as victims in five separate simulations. We calculate the mean ratio of dishonest peers in the victim's routing table among the five nodes over the course of 1,000 epochs and present the results in Table. \ref{tbl:one_victim_mean}. We observe that \OA~consistently achieves near-uniform sampling and outperforms the baseline algorithms by 4-63\% in terms of the mean ratio of dishonest peers in the victim's routing table when 5-70\% of the nodes in the network are dishonest. While \OA's performance drops when 80\% of the nodes in the network are dishonest (a scenario unlikely in a usable real-world p2p network), it still outperforms the baseline algorithms by 4\%. 

\begin{figure}[!tbp]
  \centering
  {\includegraphics[width=0.38\textwidth]{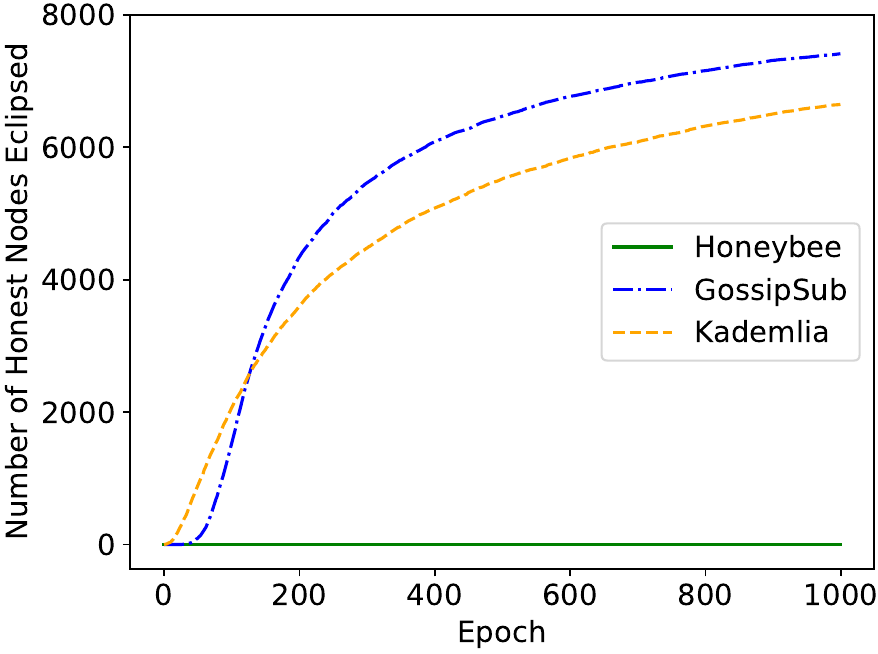}}
  \caption{All honest nodes under attack: In a network of 16,384 nodes, 50\% nodes are dishonest nodes. The honest nodes run \OA, GossipSub, and Kademlia. The dishonest nodes attempt to eclipse as many honest nodes as possible over the course of 1,000 epochs.}
  \label{fig:all_victims}
\end{figure}

\subsubsection{Attack - Multiple Victims} Fig. \ref{fig:all_victims} shows the cumulative number of honest nodes that can be eclipsed by dishonest nodes in a network with 50\% dishonest nodes in 1,000 epochs with \OA, GossipSub, and Kademlia. There are 8,192 honest nodes in total, and their goal is to eclipse (or partition, since we are addressing multiple victim nodes) as many honest nodes as possible. We observe that, in 1,000 epochs, dishonest nodes can eclipse over 75\% of honest nodes with GossipSub and Kademlia while they cannot eclipse honest nodes with \OA.

\begin{figure}[!tb]
  \centering
  \subfloat[Initialized with 62.5\% of Sybils]{\includegraphics[width=0.333\textwidth]{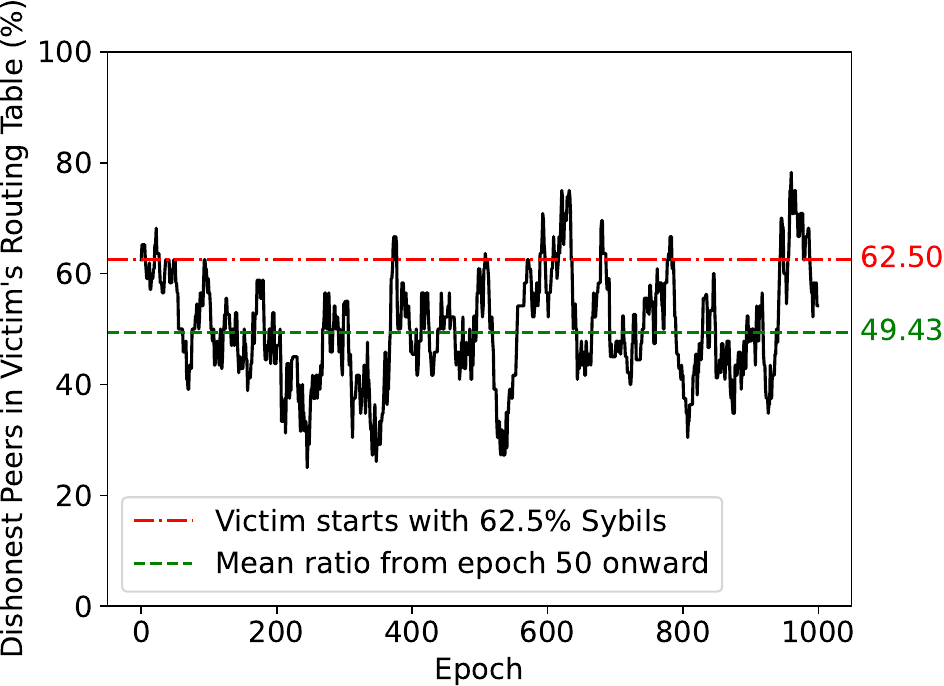}}
  \hfill
  \subfloat[Initialized with 75\% of Sybils]{\includegraphics[width=0.333\textwidth]{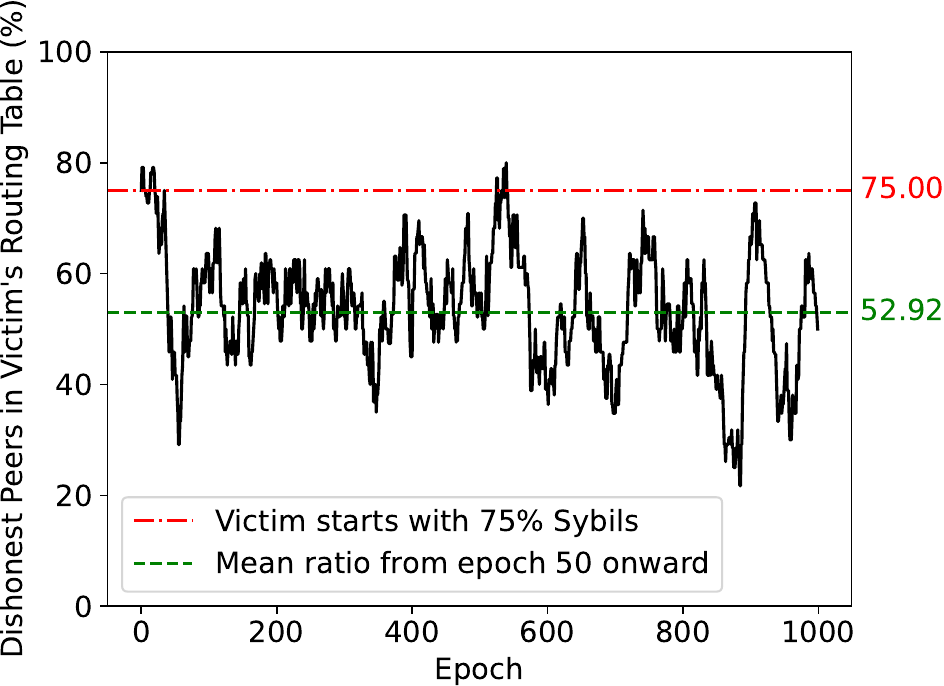}}
  \subfloat[Initialized with 87.5\% of Sybils]{\includegraphics[width=0.333\textwidth]{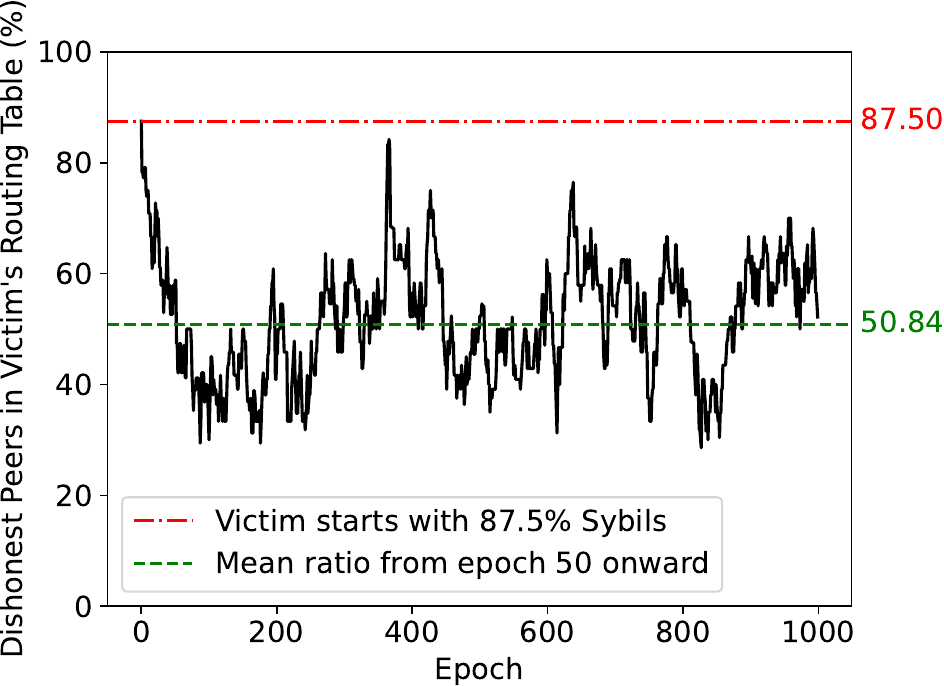}}
  \caption{An honest \OA~node attacked by 50\% of dishonest nodes in the network: when the honest node's initial routing table happens to be poorly configured.}
  \label{fig:bad_starts}
\end{figure}

\subsubsection{Bad Start} We have discussed how a \OA~node bootstraps its initial routing table configuration in \S\ref{systemModel}. It is possible for an honest \OA~node to start off with a poorly configured routing table (i.e., the ratio of dishonest nodes in one's initial routing table configuration is significantly higher than that of the entire network). To see what happens when nodes get a bad initial routing table configuration, we run experiments on \OA~networks with 50\% of dishonest nodes. In these experiments, we identified random, unfortunate \OA~nodes whose initial routing tables contain 62.5\%, 75\%, and 87.5\% of dishonest nodes. In each experiment, all the dishonest nodes in the network target and attempt to eclipse the honest node under our observation. We then examined their peer sampling performances under poorly initialized routing tables. The results are shown in Fig. \ref{fig:bad_starts}. We observe that, even with a poorly initialized routing table, \OA~nodes can still conduct near-uniform sampling, and they converge to \(\epsilon\)-uniform sampling with \(\epsilon=0.03\) quickly (within 50 epochs). Therefore, we conclude that poorly initialized routing tables do not adversely affect \OA~nodes' peer sampling performance. 

\subsubsection{Why VRW and TCC} \label{why_vrw_tcc}
\begin{figure}[!tbp]
  \centering
  \subfloat[\OA~without VRW]{\includegraphics[width=0.38\textwidth]{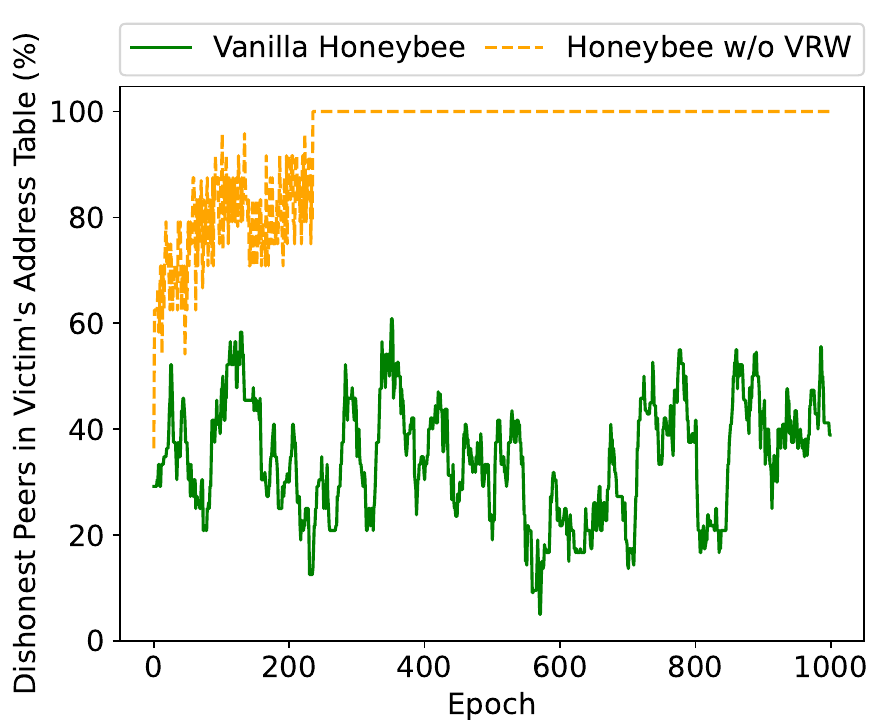}}
  \hfill
  \subfloat[\OA~without TCC]{\includegraphics[width=0.38\textwidth]{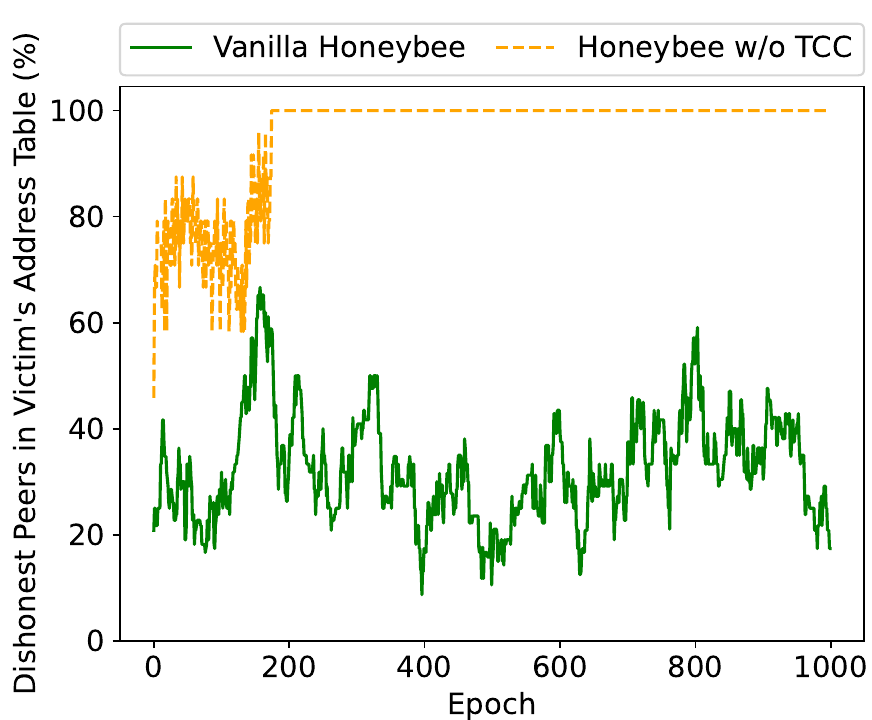}}
  \caption{\OA~without VRW or TCC for one random victim under attack: (a) displays comparison of vanilla \OA~and \OA~without random walk verification mechanism; (b) displays comparison of vanilla \OA~and \OA~without table consistency check mechanism.}
  \label{fig:Without_VRW_or_TCC}
\end{figure}

We demonstrate the importance of verifiable random walks and table consistency checks by showing the reader what would happen if we did not use verifiable random walks or table consistency checks in \OA. Without the random walk verification mechanism, a random walk is not necessary random since dishonest nodes can hijack it by reconfigure their address tables according to their preferences. Similarly, without the table consistency check mechanism, dishonest nodes are able to store multiple copies of address tables and use different copies to route random walks according to their preferences. In other words, without verifiable random walks or table consistency checks, dishonest nodes can carry out a de facto adversarial routing attack. 

We simulate two altered versions of \OA: one without the random walk verification mechanism, the other without the table consistency check mechanism. Fig. \ref{fig:Without_VRW_or_TCC}a shows the process of dishonest nodes target and eclipse a random honest victim node when we remove the verification mechanism from random walks. Fig. \ref{fig:Without_VRW_or_TCC}b shows the process of dishonest nodes target and eclipse a random honest victim node when we remove table consistency checks. In both scenarios, all nodes start with a random address table configuration, and we have 30\% of dishonest nodes in the network. We compare them with the unaltered \OA~under the same setting to present the consequences. 

\subsection{Practical Considerations} \label{practical_consider}
So far, we have demonstrated \OA's effectiveness and efficiency with analyses and experiments. Now, we would like to briefly discuss some practical considerations when it comes to applying \OA~in real-world networks. 


In \S\ref{systemModel}, we have mentioned that \OA~can use block headers as its source of randomness. However, what if the block headers are not secure (e.g., they may be susceptible to tampering, such as grinding attacks~\cite{gitbookProtocolLayer}) in the blockchain network we are working on? To address this issue, we can introduce an extra layer of protection to the consensus mechanism using a verifiable delay function (VDF)~\cite{boneh2018verifiable}, which ensure that randomness is both secure and verifiable. Moreover, we may want to apply \OA~to general p2p networks which do not have a blockchain. As we have discussed in \S\ref{OA}, \OA~needs a shared source of randomness for nodes to perform decentralized peer sampling. This source of randomness does not necessarily have to be derived from a blockchain, as long as it is secure and resistant to manipulation. A good example of such a secure source of randomness is decentralized randomness beacons~\cite{cloudflareDistributedRandomness}, which are designed to generate unbiased and publicly verifiable randomness in a decentralized manner. By incorporating a reliable shared source of randomness, we can extend the applicability of \OA~beyond blockchain-based systems to a broader range of p2p networks. 

Note that the frequency of shared randomness generation does not impact the frequency or time consumption of random walks and peer samplings. The reasons are as follows: (1) Popular blockchains (e.g., Ethereum~\cite{etherblocktime}) and decentralized randomness beacons (e.g., Drand~\cite{drandDrandDistributed}) are able to provide shared randomness within tens of seconds. (2) As mentioned in \S\ref{vrw}, in each round of \OA, a fraction \(\eta\) of nodes become eligible and perform peer sampling. Even if the time required to generate fresh random values is a lot longer (e.g., at least twice as long, which is uncommon) than the typical time required to perform random walks, we can generate extra pseudo-random values using amplification~\cite{colbeck2012free}. By incorporating an amplification process, more than one round of peer sampling can be conducted. 
 
As for communication overhead, in a \OA~network with \(n\) nodes, a random walk consists of at most \(O(log(n))\) hops, which correspond to \(O(log(n))\) messages. Suppose a node's routing table size is \(k\), a node only needs to exchange \(O(klog(n))\) messages to update its entire routing table. The communication overhead for all nodes in the network to update their entire routing tables is just \(O(knlog(n))\). Regarding the time required for all nodes to update their entire routing tables, since in each round we allow a fraction \(\eta\) of nodes to perform peer sampling, the network requires \(\frac{k}{2\eta}\) rounds for all nodes to update their entire routing tables. To calculate the wall-clock time, assuming \(k=24\), \(\eta=0.1\) and a round costs 12 seconds~\cite{etherblocktime}, the total time consumption would be \(24\) minutes. Even if we relax the assumption by setting \(\eta=0.001\) (i.e, only \(0.1\%\) of nodes can conduct peer sampling in a round), the network would still allow all nodes to update their entire routing tables within \(1.67\) days, which should be more than acceptable in most p2p network applications. Further improving the efficiency of \OA~peer sampling is an interesting direction for future work.

In terms of adoption, networks employing other p2p protocols can smoothly transition into Honeybee in a hybrid manner. During the transition, the network encourages nodes to adapt to Honeybee while allowing the use of other protocols. Honeybee nodes work normally with each other as we described in \S\ref{OA}. In addition, Honeybee nodes store the addresses of non-Honeybee nodes in a separate table to stay in touch with them. Non-Honeybee nodes can choose to adapt to Honeybee at their own discretion.


\section{Related Work}
We have covered some related work on \OA~in \S\ref{sec:intro}. In this section, we will discuss additional related work.

Researchers have made extensive effort to improve the security and efficiency of p2p networks. Protocols like Castro et al. \cite{castro2002secure}, S/Kademlia~\cite{4447808}, SecureCyclon~\cite{10272533}, and RAPTEE~\cite{9912135} enhance the security of their underlying DHT and Gossip networks by restricting free node ID generation through a centralized authority, cryptographic puzzles, or trusted execution environments. In the Generals' Scuttlebutt \cite{10.1145/3548606.3560638}, Coretti and colleagues designed a Byzantine-resilient gossip protocol under the proof-of-stake setting. Under the protocol, nodes build a connected backbone of high-staking nodes. However, light nodes (or nodes with small stakes) are easily sacrificed in this protocol, and stake distribution is not always readily available in a completely decentralized environment. In Brahms \cite{bortnikov2008brahms}, Bortnikov and colleagues presented a Byzantine resistant sampling mechanism in the gossip model. The mechanism, however, relies on the assumption that the ability of malicious nodes to send messages is throttled. In AnonBoot~\cite{10.1145/3320269.3384729}, the authors presented a peer configuration protocol that relies on a proof-of-work chain and requires nodes to advertise themselves on-chain using cryptographic puzzles. Traditional P2P system proposals, such as Tarzan \cite{10.1145/586110.586137} and MorphMix \cite{rennhard2004practical}, handled attacks based on the assumption that attackers have limited resources (e.g., IP addresses) and that only contiguous ranges of IP addresses are under malicious control. Wang and Borisov proposed Octopus \cite{6258005}, a secure lookup strategy in p2p networks, and the strategy relies on a centralized authority to remove malicious nodes from networks. Panchenko and colleagues proposed NISAN~\cite{10.1145/1653662.1653681}, a protocol that conducts anonymous lookups for Chord DHTs using aggregated greedy search and search value hiding. The protocol, however, is vulnerable to information leak attacks and does not offer sufficient protection against Sybil attacks~\cite{10.1145/1866307.1866343, 10.1145/3320269.3384729}. Recent researches inspired by the multi-armed bandit problem \cite{zhang2023kadabra,10.1145/3382734.3405704} attempted to reduce the communication latency in p2p networks without losing security by learning from the performance of neighboring peers and update peer selections based on the knowledge. In Kiffer et al. \cite{kiffer2021under}, the authors studied the connectivity and block propagation mechanism of Ethereum's p2p network. In Vedula et al. \cite{10174933}, the authors formalized the p2p topology construction in Ethereum as a game between miners. In Król et al. \cite{krol2023data}, the authors conducted an analysis on p2p networking requirements for data availability sampling in Ethereum. 

Random walks in p2p networks have been well-discussed by researchers since the 2000s. In contrast to the reason for \OA~employing verifiable random walks, a plethora of works on p2p random walks primarily focused on resource searching, load balancing, and topology formation for p2p networks \cite{bisnik2005modeling,gkantsidis2004random,zhong2008convergence,navaz2015load,kwong2005relationship}. In ShadowWalker~\cite{10.1145/1653662.1653683}, nodes achieve anonymous communication through random walks with neighborhood verification from shadow nodes, but it relies on the assumption that Sybil nodes constitute at most 20\% of the network and possess limited capabilities. In Bisnik et al. \cite{bisnik2007optimizing}, the authors improved resource searching efficiency with adaptive random walks that take advantage of the feedback from previous searches. In Massoulié et al. \cite{massoulie2006peer}, the authors approached the problem of estimating the number of peers in a p2p network with random walks. In Das Sarma et al. \cite{sarma2015efficient}, the authors attempted to perform node sampling with minimized round complexity and message complexity using continuous random walks under limited adversarial behaviors. In the Dandelion Bitcoin network protocol \cite{10.1145/3084459}, the authors applied random walks before message broadcast to protect the anonymity of the transaction parties. Augustine et al.~\cite{10.1145/3490148.3538588} use multiple random walks to achieve peer sampling for DHTs by introducing a rather complicated process of building and dismantling hypercubes of nodes. However, their random walks can be susceptible to tampering and only allow a maximum of \(\frac{n}{polylog(n)}\) Byzantine nodes in a network of \(n\) nodes. 

\OA~can be applied to enhance blockchain data availability sampling at the network layer through its node sampling capability. Plenty of efforts have been made in blockchain layer-2 scaling \cite{scaling,rollups,10.5555/3388242.3388299,10302968, ma2025sok}. As a critical component of layer-2 solutions, data availability sampling itself has many important research problems \cite{Neu_2022}. In Hall-Andersen et al. \cite{cryptoeprint:2023/1079}, the authors initiated a cryptographic study of data availability sampling and demonstrated its relation with erasure codes. In Al-Bassam et al. \cite{al2021fraud}, the authors developed a complete fraud and data availability proof scheme in which they claimed that light nodes can have a security guarantee close to the level of a full node under certain assumptions. In Yu et al. \cite{yu2020coded}, the authors created coded Merkle trees to improve the protection of light nodes against data availability attacks. In Cao et al. \cite{9284704}, the authors presented a decentralized collaborative light-node-only verification mechanism in which light nodes can conduct block verification without the help of a full node. In addition to the advancement in security, there are some researches that focus on lowering the communication and computation overhead on nodes in the data availability setting \cite{sheng2021aced,kiayias2020non,bunz2020flyclient,kiayias2016proofs}. 

\section{Conclusion}
We presented Honeybee, a decentralized p2p sampling algorithm that achieves near-uniform peer sampling. Unlike existing algorithms, Honeybee is the first that conducts Byzantine-tolerant genuine random walks in networks, using verifiable random walks and table consistency checks, for node discovery and sampling. Honeybee is secure against various adversarial strategies under diverse settings and has broad applications in blockchains and p2p networks. In our experiments, we observe that Honeybee consistently achieves \(\epsilon\)-uniform sampling with \(\epsilon=0.03\). An interesting future research topic would be to generate a time-inhomogeneous Markov chain using Honeybee nodes and their address tables. Then, we could examine from a theoretical perspective whether mixing occurs and the potential mixing time when all nodes actively perform random walks. 

\bibliographystyle{ACM-Reference-Format}
\bibliography{main}

\appendix

\section{Outgoing Address Table Update} \label{alg1_details}
In Algorithm \ref{alg2}, we provide a more detailed version of Algorithm \ref{alg1} and illustrate the mechanism of \OA's iterative verifiable random walks.
\begin{algorithm}[!tbp]
\DontPrintSemicolon
\SetKwInOut{Input}{input}\SetKwInOut{Output}{output}
\Input{peers $\Gamma^v_\text{curr}$ in outgoing address table of current epoch; 
shared pseudo-randomness seed $\mathcal{R}$;
$v$'s secret key $SK_v$;
epoch counter $\varrho_\text{curr}$;
system-defined address table inconsistency threshold $\tau$}
\Output{updated set of peers $\Gamma^v_\text{next}$ for next epoch}
\If {$\mathcal{R}$ indicates $v$ is a currently eligible node} {
\tcc{$v$ infers the random walk path length $p$ from $\mathcal{R}, v$.}
$p_c \leftarrow 0 $ \tcc{Hop counter of the random walk}
$u^* \leftarrow v $ \tcc{$u^*$ stores node on current hop of the random walk} 
$d \leftarrow -1 $ \tcc{Index of node on first hop of the random walk} 
\While {$p_c < p$} {
\tcc{$\textsc{VRF}$() returns a pseudo-random number $i$ and its corresponding proof $\pi_v$} 
$(i, \pi_v) \leftarrow VRF(\mathcal{R}, \varrho_\text{curr}, p_c, u^*, SK_v)$ \;
\tcc{$\Gamma^{u^*}_\mathrm{curr}(i)$ returns the i-th node in $\Gamma^{u^*}_\mathrm{curr}$ and its address table snapshot} 
$(u^*, \omega) \leftarrow \Gamma^{u^*}_\mathrm{curr}(i)$ \;
$p_c \leftarrow p_c + 1$ \;
\If {$p_c = 1$}{
$d \leftarrow i$
}
\tcc{$\textsc{RequestAddrTable}(u^*)$ returns an address table that node $u^*$ claims to be its latest}
$\omega^* \leftarrow \textsc{RequestAddrTable}(u^*)$ \;
\If {$\textsc{Diff}(\omega, \omega^*) > \tau$} {
\tcc{$v$ issues a fraud proof against $u^*$ with $\omega, \omega^*$}
}
\tcc{$v$ also receives other address table snapshots $\omega'$ of $u^*$ when interacting with other nodes}
\If {$\textsc{Diff}(\omega', \omega^*) > \tau$} {
\tcc{$v$ issues a fraud proof against $u^*$ with $\omega', \omega^*$}
}
}
$\Gamma^v_\text{next} \leftarrow \Gamma^v_\text{curr} \backslash \{ \Gamma^v_\text{curr}(d) \} \cup \{ u^* \}$\; 
}
\caption{A detailed version of Algorithm \ref{alg1} - algorithm outline for updating entries of outgoing address table of node $v$ in each epoch.}
\label{alg2}
\end{algorithm}

\section{Analysis} \label{theory} 
We carry out an analysis to better understand the effectiveness of \OA~from a theoretical perspective. We first discuss the effectiveness of \OA's protocol enforcement. Then, we introduce the theoretical model and present the analysis results. The model we use for the analysis is simplified (compared to \S\ref{systemModel}) for analytical tractability. 

\subsection{Enforcement of \OA} \label{oa_enforcement}
Recall that, in \OA, with verifiable random walks (VRWs), a random walk that arrives at (or terminates at) an honest node must be a genuine random walk. While the random walk has to be genuine, a dishonest node may use multiple (different) routing tables and share them with different nodes based on its needs. By manipulating multiple routing tables, a dishonest node can carry out de facto rerouting on a random walk while pretending not to tamper with its genuineness. As discussed in \S\ref{tcc} and \S\ref{attack_strategy}, we call the act of manipulating multiple routing tables an equivocal table attack.

An equivocal table attack poses a critical threat to peer sampling algorithms, as it not only provides misleading information to honest nodes but also leads to further attacks that impact honest nodes' neighborhoods (e.g., adversarial routing and eclipse attack). As we have addressed in \S\ref{tcc}, in \OA, we catch, slash, and deter such attacks using table consistency checks (TCCs). Here, we would like to examine the probability that a dishonest node carrying out an equivocal table attack gets verifiably detected and slashed by TCCs. 

In a \OA~network where the routing table size of nodes and the length of random walks are \(O(log(n))\), suppose a dishonest node \(u\) carries out an equivocal table attack, and it has $m$ different routing tables $\Gamma^u_\text{1}, \Gamma^u_\text{2}, \ldots, \Gamma^u_m$. 
Let us call the round in which \(u\) starts the equivocal table attack round \(1\). 
In round \(1\), \(u\) shares $\Gamma^u_\text{1}$ to a nonempty set of honest nodes, $\Gamma^u_\text{2}$ to another nonempty set of honest nodes, and so on. 
In the following analysis, we show that Honeybee's TCCs can catch equivocal attacks with high probability if the number of different routing tables $m$ is greater than $\sqrt{n}$, where $n$ is the number of nodes. 

\smallskip 
\noindent 
{\bf Model, assumptions, and notation.}
We let $k$ denote the size of the address table at each node. 
$\eta$ is the fraction of nodes eligible to do a random walk in each round.  
$l$ denotes the maximum length of a random walk conducted by any node. 
In practice, we can choose both $k$ and $l$ as $O(\log n)$. 
$b$ is the number of routing tables stored in the scratch pad. 
We set $b$ to be sufficiently large that all nodes encountered by a node in $k/(2\eta)$ consecutive rounds can be stored in the scratchpad. 

Notice that an equivocating adversarial node with $m$ routing tables, can simply be thought of as $m$ different nodes in the network.
Assuming such a model, we let the network be a random regular graph with a diameter of $O(\log n)$ at each round. 
Here $n$ is the total number of nodes in the network including the $m$ different adversarial nodes. 
We assume the local neighborhood around each node is tree like.  
Each round a fraction of $\eta$ nodes are eligible to conduct the random walk. 
We assume the sets of eligible nodes in $1/\eta$ consecutive rounds are disjoint. 
$\tilde{n} < n$ is the number of honest nodes in the network, and $v_1, v_2, \ldots, v_{\tilde{n}}$ denotes the honest nodes. 
Apart from carrying multiple routing tables, we assume the adversarial nodes otherwise follow the Honeybee protocol. 

\smallskip
\noindent
{\bf Problem formulation.}
Define a window as $k/(2\eta)$ consecutive rounds.
For any node (including adversarial ones), after $k/\eta$ rounds, its address table is likely to be completely refreshed with fresh addresses in Honeybee. 
By defining the window as half of the time taken to refresh the routing table, two honest nodes that encounter an adversarial node during one window can  compare their versions of the adverarial node's routing table in the next window if the nodes happen to meet (i.e., one node visits the other during a random walk). 
If the versions differ significantly (e.g., more than half of the addresses differ), then the two nodes can issue a fraud proof against the equivocating adversarial node.  
In the following, we are interested in the probability and time it takes for honest nodes to detect the presence of an equivocating adversarial node with $m$ different routing tables. 
We assume that dishonest nodes cooperate with each other and do not contribute to TCCs among honest nodes. 
While they do not contribute to TCCs, dishonest nodes cannot help themselves or each other evade TCCs because the snapshots of tables are digitally signed. 

\smallskip
\noindent 
{\bf Analysis.}
Let $u_1, u_2, \ldots, u_m$ be the $m$ adversarial nodes. 
Consider node $u_1$. 
For each honest node $v_1, v_2, \ldots, v_{\tilde{n}}$, let $\mathbf{1}_{v_i \rightarrow u_1}^t$ be the random variable that is 1 if node $v_i$ visits node $u_1$ during the current round $t$, and 0 otherwise.
Each honest node conducts its walk independently of other honest nodes; therefore, the random variables $\mathbf{1}_{v_1 \rightarrow u_1}^t, \mathbf{1}_{v_2 \rightarrow u_1}^t, \ldots, \mathbf{1}_{v_{\tilde{n}} \rightarrow u_1}^t$ are mutually independent. 
Let 
\begin{align}
    X_1 = \sum_{t=1}^{k/(2\eta)} X_1^t = \sum_{t=1}^{k/(2\eta)} \sum_{i=1}^{\tilde{n}} \mathbf{1}_{v_i \rightarrow u_1}^t,  
\end{align}
be the number of visits of $u_1$ by honest nodes in the window.  
We similarly define $X_i$ as the number of visits of adversarial node $u_i$ for $i=1,2,\ldots,m$ by honest nodes during the window.
\begin{lemma}
\label{lem:X lower bound}
For $i=1,2,\ldots,m$, $X_i \geq 1$ with probability $1-o(1)$. 
\end{lemma}
\begin{proof}
Consider the local neighborhood around node $u_1$ at round $t$.
By assumption, this neighborhood is a tree of degree $k$.
Some of the nodes in this tree could be other malicious nodes (out of the $m$ malicious nodes). 
For an honest node $v$ at a distance of $i$ hops away from $u_1$, the probability of $v$ not reaching $u_1$ during round $t$ (if $v$ is eligible at $t$) is at most $1-1/k^i$.

To decrease the likelihood of any honest node reaching $u_1$ with its random walk, the worst possible configuration is for the $m-1$ adversarial nodes to be situated as close to $u_1$ as possible.
That is, roughly speaking, the $m$ nodes all form a ball (of radius roughly $\log_k (m)$ around $u_1$).\footnote{This configuration is used only to simplify our exposition. The conclusions hold true for any other configuration as well.} 
Therefore there exists an $O(\log n) > i > \log_k(m)$ such that all nodes $i$ hops away from $u_1$ are honest.  
Let $V_i^t$ denote these honest nodes for round $t$. 
Considering only these honest nodes, we have 
\begin{align}
P(X_1 = 0) \leq P(\mathbf{1}_{v \rightarrow u_1}^t = 0 ~\forall v \in V_i^t, \forall t \in \{1,\ldots, k/(2\eta) \})    \\
\leq ((1 - 1/k^i)^{k(k-1)^{i-1}})^{k/2} 
\leq (e^{-k(k-1)^{i-1}/k^i})^{k/2} \leq (e^{-(k-1)^{i}/k^i})^{k/2} 
= o(1) 
\end{align}
\end{proof}

Let $\mathcal{X}_i$, for $i=1,2,\ldots,m$ denote the set of honest nodes that visit $u_i$ during the first window. 
In the following, we bound the probability of honest nodes detecting an equivocal attack in the subsequent window. 

\begin{theorem}
\label{thm: equivocation}
If $m = \Omega(\sqrt{n})$, with probability greater than $1-1/e$, the equivocation attack is detected by the honest nodes.  
\end{theorem}
\begin{proof}
Consider any node $v \in \mathcal{X}_1$. 
Let $\mathcal{R}_v^t$ denote the set of all possible $l$-length random walks for $v$. 
Each walk $w \in \mathcal{R}_v^t$ occurs with a probability of $1/k^l$. 
By assumption, for any node $v' \in \mathcal{X}_2 \cup \mathcal{X}_3 \cup \ldots \cup \mathcal{X}_m$ there exists a walk $w \in \mathcal{R}_v^t$ with $v'$ visited on the walk $w$. 
It follows that there must be at least $(X_2 + X_3 + \ldots + X_m)/l$ walks in $\mathcal{R}_v^t$ such that each walk contains at least one node of $\mathcal{X}_2 \cup \mathcal{X}_3 \cup \ldots \cup \mathcal{X}_m$.
Here, we consider the case where $\mathcal{X}_1, \mathcal{X}_2, \ldots, \mathcal{X}_m$ are disjoint sets. 
For otherwise, the equivocating node is trivially detected by an honest node in the intersection of two non-disjoint sets $\mathcal{X}_i$ and $\mathcal{X}_j$. 
The bound we derive below for probability of missed detection holds trivially true in this case. 
Therefore, the probability of $v$ not visiting any of the nodes in $\mathcal{X}_2 \cup \mathcal{X}_3 \cup \ldots \cup \mathcal{X}_m$ during its walk is bounded as
\begin{align}
P(v \text{ does not visit any node in } \mathcal{X}_2 \cup \ldots \cup \mathcal{X}_m \text{ at round } t) \leq 1 - (X_2 + \ldots + X_m)/(l k^l).
\end{align} 
Since a fraction $\eta$ of the nodes in $\mathcal{X}_1$ are eligible during a round and the walks are conducted independently, we have 
\begin{align}
P(\text{no node in } \mathcal{X}_1 \text{ visits any node in } \mathcal{X}_2 \cup \ldots \cup \mathcal{X}_m \text{ at round } t) \leq \left( 1 - (X_2 + \ldots + X_m)/(l k^l) \right)^{\eta X_1}.    
\end{align}
Further, there are $k/(2\eta)$ rounds in a window. 
Since the randomness across rounds is independent, the probability that no node in $\mathcal{X}_1$ visits any node in $\mathcal{X}_2 \cup \ldots \cup \mathcal{X}_m$ during the rounds in the window can be bounded as
\begin{align}
P(\text{no node in } \mathcal{X}_1 \text{ visits any node in } \mathcal{X}_2 \cup \ldots \cup \mathcal{X}_m \text{ in the window }) \leq \notag \\
\left( 1 - (X_2 + \ldots + X_m)/(l k^l) \right)^{\eta X_1 k/(2\eta)} \\
\leq e^{-(X_2 + \ldots + X_m) \eta X_1 k / (2\eta l k^l)} \leq e^{-X_1 (X_2 + \ldots + X_m)/(2l k^{l-1})}.        
\end{align}
Since the walks conducted by nodes in the sets $\mathcal{X}_1, \ldots, \mathcal{X}_m$ are independent, we have 
\begin{align}
P(\text{no pair of nodes in } \mathcal{X}_1 \cup \ldots \cup \mathcal{X}_m \text{ visit each other in the window}) \leq \notag \\
\prod_{i=1}^m e^{-X_i (X_1 + + \ldots + X_m - X_i)/(2l k^{l-1})}.        
\end{align}
If $X_1, X_2, \ldots, X_m \geq 1$ with high probability (Lemma~\ref{lem:X lower bound}), 
\begin{align}
P(\text{no pair of nodes in } \mathcal{X}_1 \cup \ldots \cup \mathcal{X}_m \text{ visit each other in the window}) \leq 
e^{-m(m-1)/(2l k^{l-1})}. 
\end{align}
The probability of missed detection is less than $1/e$ if 
\begin{align}
 m > \sqrt{2l k^{l-1}} + 1 = \Theta(\sqrt{n})
\end{align}
\end{proof}
Thus, if there are more than $\sqrt{n}$ equivocating nodes, they can be detected with high probability.

\subsection{Sampling Performance}

Next, we analyze the sampling efficiency achieved by Honeybee. 
It is easy to see that if only a single node performs a random walk on a regular graph, it can eventually (after mixing) achieve uniform sampling. 
However, in Honeybee, multiple nodes concurrently perform walks. 
Moreover, the walk of one Honeybee node can cause a node to lose its outgoing edge to a previously sampled neighbor. 
A node's random walk may also not succeed in gaining a new outgoing neighbor for the node if too many other random walks terminate at the same neighbor node. 
We seek to understand the effective rate at which a Honeybee node can sample new peers, in a network where multiple nodes are simultaneously engaged in sampling at any round. 

\smallskip
\noindent 
{\bf Model, assumptions, and notation.}
We extend the model used in \S\ref{oa_enforcement}. 
Time is divided into rounds. 
At the end of round $t$, let $M_u(t)$ be the out-address table of node $u$ for any $u \in V$. 
We assume $M_u(t)$ is limited to at most $k$ entries for any $u \in V$ at all times. 
Each node also has an in-address table. 
The in-address table for node $u \in V$ at the end of round $t$ is denoted by $I_u(t)$. 
Each round a subset of nodes $E(t)$ become eligible to conduct a random walk. 
The number of eligible nodes satisfies $|E(t)| = \eta n$, where $\eta > 0$ is a parameter. 
The eligibility of nodes follows a periodic schedule: $E(t) = E(t + 1/\eta)$ for all time $t\geq 0$.
Also, $\cup_{t' = t}^{t+1/\eta} E(t) = V$ for all $t \geq 0$.  
Time is further divided into epochs, where each epoch lasts for $1/\eta$ rounds (assumed to be an integer). 
Let $e(t) = \lfloor t\eta \rfloor$ be the epoch number at round $t$. 

If a node $u$ is eligible at round $t$, it performs a random walk to update the $(e(t)\bmod k)$-th entry of its out-address table. 
When the random walk terminates (say, at node $v$), the initiator of the walk $u$ makes a sampling request to $v$. 
If $v$ accepts the request, $u$ adds $v$ at the $(e(t)\bmod k)$-th index in its out-address table, while $v$ adds $u$ within its in-address table (the index where it is added in $v$'s in-address table does not matter). 

We assume a node $v$ tries to accept all incoming sampling requests at round $t$ as follows.
At the end of round $t$, after all eligible nodes have terminated their respective random walks, node $v$ removes incoming edges $(u', v)$ such that: (1) $u' \in E(t)$, and (2) $M_{u'}^{e(t)}(t) = v$. 
Let $J_v(t) \subseteq I_v(t)$ be the set of incoming edges at $v$ that are leftover. 
Next, suppose $v$ receives $i\geq 0$ sampling requests. 
If $i \geq k$, node $v$ chooses $k$ requests out of the $i$ requests randomly, and admits them. 
Node $v$ deletes all $J_v(t)$ existing incoming edges. 
If $i < k$, node $v$ accepts all $i$ requests. 
It deletes $\max(0, i-k+|J_v(t)|)$ incoming edges randomly from $J_v(t)$. 

For any $u \in V, i \in \{0,1,\ldots, k-1 \}$ and $t\geq 0$, $M_u^i(t) \in V \cup \{\phi\}$. 
Here $M_u^i(t) = \phi$ means that node $u$ does not have a valid outgoing edge at the $i$-th index in its out-address table at time $t$. 
Note that it is possible for $M_u^i(t) = u$ (i.e., a self-loop) and also for $M_u^i(t) = M_u^{j}(t)$ for $i\neq j$ (i.e., multiple out-edges pointing at the same neighbor). 

If $u \in E(t)$ and $i = e(t)\bmod k$, node $u$ conducts a random walk to update $M_u^i$. 
We assume the random walk terminates at a node sampled uniformly in $V$. 
The randomness used for this sampling is independent of all other randomness in the system. 

\smallskip 
\noindent
{\bf Problem formulation.}
We aim to evaluate Honeybee's performance using two metrics: sample success and sample retention. 
Sample success is the likelihood that an eligible node successfully updates its out-address table during a round. 
Formally, for a node $u$ at time $t$ such that $u \in E(t)$, the sampling success for $u$ at time $t$ is defined as 
\begin{align}
\sigma_{u, t} = \mathbb{E}[\mathbf{1}_{M_u^{e(t)\bmod k}(t) \neq \phi}].
\end{align}
Over a time horizon of $T$ rounds, the average number of successful samples obtained by $u$ per epoch is $\sum_{t=0: u \in E(t)}^{T-1} \sigma_{u,t}/\sum_{t=0:u\in E(t)}^{T-1}1$. 
When a node is successful in sampling during a round, the sample retention metric measures the time the node retains the sample without being evicted by the sampled node. 
Ideally, a newly obtained sample must be retained for $k$ epochs ($k/\eta$ rounds); however, as discussed above, a node can delete some of its incoming edges to make room for new sampling requests over time. 
For a node $u$ that is eligible at round $t$, we define $u$'s sampling retention at $t$ as
\begin{align}
\rho_{u,t} = \mathbb{E}\left[ \sum_{t'= t: u\in E(t)}^{t+k/\eta-1} \mathbf{1}_{M_u^{e(t)\bmod k}(t') \neq \phi} \right].  
\end{align}

\smallskip 
\noindent
{\bf Analysis.}
Our main result is the following theorem. 
\begin{theorem}
For any node $u \in V$ and round $t$ such that $u \in E(t)$ we have 
\begin{align}
    \sigma_{u,t} \geq 1 - \frac{\eta - 1/n}{k}, \\
    \rho_{u,t} \geq \sigma_{u,t}\frac{1-e^{-k}}{1-e^{-\eta}}. 
\end{align}
\end{theorem}
\begin{proof}
We first evaluate the sampling success of a node. 
Consider a node $u$ at round $t$ such that $u \in E(t)$. 
Let $\mathcal{E}_u^v(t)$ be the event that $u$'s random walk terminates at $v$ at time $t$.  
We have 
\begin{align}
\sigma_{u, t} &= P(M_u^{e(t) \bmod k}(t) \neq \phi) \notag \\
&= \sum_{v \in V} P(\mathcal{E}_u^v(t)) P(M_u^{e(t) \bmod k} (t) \neq \phi | \mathcal{E}_u^v(t)) \notag \\
&= \sum_{v \in V} \frac{1}{n} P(M_u^{e(t) \bmod k} (t) \neq \phi | \mathcal{E}_u^v(t)), \label{eq:succsamp1}
\end{align}
since we have assumed each node samples a node uniformly at random. 
Next, let $\mathcal{F}_u^{j,v}(t)$ be the event that $j$ other eligible nodes, for $j=0,1,\ldots,\eta n-1$, also terminate their respective random walks at $v$ in round $t$.   
We have 
\begin{align}
P(M_u^{e(t) \bmod k} (t) \neq \phi | \mathcal{E}_u^v(t)) = \sum_{j=0}^{\eta n-1} P(\mathcal{F}_u^{j,v}(t) | \mathcal{E}_u^v(t)) P(M_u^{e(t) \bmod k} (t) \neq \phi | \mathcal{F}_u^{j,v}(t), \mathcal{E}_u^v(t)).  \label{eq:succsamp2}
\end{align}
Since each eligible node performs its sampling independently of the other nodes, the events $\mathcal{F}_u^{j,v}(t)$ and $\mathcal{E}_u^v(t)$ are independent. By assumption, each node samples a node uniformly at random. 
Therefore, 
\begin{align}
P(\mathcal{F}_u^{j,v}(t) | \mathcal{E}_u^v(t)) = P(\mathcal{F}_u^{j,v}(t)) = \binom{\eta n-1}{j}\frac{1}{n^j}\left(\frac{n-1}{n}\right)^{\eta n-1-j}. \label{eq:succsamp3}
\end{align}
When $v$ receives a set of incoming sampling requests, it tries to admit as many of these requests. 
Specifically, if $v$ receives fewer than or equal to $k$ total requests, all of the requests are admitted; and if $v$ receives more than $k$ requests a random subset of these $k$ requests are admitted. 
Therefore, 
\begin{align}
P(M_u^{e(t)\bmod k}(t) \neq \phi | \mathcal{F}_u^{j, v}(t), \mathcal{E}_u^v(t)) &= 
\begin{cases}
1 & \text{for } j=0,1,\ldots,k-1 \\
\binom{j}{k-1}/\binom{j+1}{k} & \text{for } j=k,k+1,\ldots, \eta n - 1 
\end{cases} \notag \\
&= \begin{cases}
1 & \text{for } j=0,1,\ldots,k-1 \\
k/(j+1) & \text{for } j=k,k+1,\ldots, \eta n - 1 
\end{cases}. \label{eq:succsamp4}
\end{align}
Substituting equations~\eqref{eq:succsamp2}, \eqref{eq:succsamp3}, and \eqref{eq:succsamp4} in equation~\eqref{eq:succsamp1}, we get 
\begin{align}
\sigma_{u,t} =& \sum_{v \in V}\frac{1}{n} \left[ \sum_{j=0}^{k-1} \binom{\eta n-1}{j}\frac{1}{n^j}\left(\frac{n-1}{n}\right)^{\eta n-1-j} + \sum_{j=k}^{\eta n-1} \binom{\eta n-1}{j}\frac{1}{n^j}\left(\frac{n-1}{n}\right)^{\eta n-1-j} \frac{k}{j+1} \right] \notag \\
=& \sum_{j=0}^{k-1} \binom{\eta n-1}{j}\frac{1}{n^j}\left(\frac{n-1}{n}\right)^{\eta n-1-j} + \sum_{j=k}^{\eta n-1} \binom{\eta n-1}{j}\frac{1}{n^j}\left(\frac{n-1}{n}\right)^{\eta n-1-j} \frac{k}{j+1}.   \label{eq:succsamp5}
\end{align}
Using the Markov inequality, equation~\eqref{eq:succsamp5} can be lower bounded as 
\begin{align}
\sigma_{u,t} \geq 1 - \frac{\eta-1/n}{k}.    
\end{align}

Next, we evaluate the sample retention metric for Honeybee. 
As before, we consider a node $u$ that is eligible at time $t$. 
The probability that $u$ retains its out-edge at time $t+1$ is given by 
\begin{align}
P(M_u^{e(t)\bmod k}(t+1) \neq \phi) = P(M_u^{e(t)\bmod k}(t+1) \neq \phi, M_u^{e(t)\bmod k}(t) \neq \phi)  \label{eq:retentsamp0}  \\
= P(M_u^{e(t)\bmod k}(t) \neq \phi) P(M_u^{e(t)\bmod k}(t+1) \neq \phi | M_u^{e(t)\bmod k}(t) \neq \phi), \label{eq:retentsamp1}
\end{align}
where equation~\eqref{eq:retentsamp0} follows because the event $M_u^{e(t)\bmod k}(t+1) \neq \phi$ implies the event $M_u^{e(t)\bmod k}(t) \neq \phi$. 
The first term in the above product, i.e., the probability that $u$'s sampling at time $t$ results in a success is given by equation~\eqref{eq:succsamp5}. 
Let $\mathcal{G}_u^v(t+1)$ be the event that none of the eligible nodes at time $t+1$ make a sampling request to node $v$. 
We can now lower bound the second term in equation~\eqref{eq:retentsamp1} as 
\begin{align}
P(M_u^{e(t)\bmod k}(t+1) \neq \phi | M_u^{e(t)\bmod k}(t) \neq \phi) = \notag \\ \sum_{v\in V} P(M_u^{e(t)\bmod k}(t) = v | M_u^{e(t)\bmod k}(t) \neq \phi) P(M_u^{e(t)\bmod k}(t+1) \neq \phi | M_u^{e(t)\bmod k}(t) = v) \geq  \notag \\
\sum_{v\in V} P(M_u^{e(t)\bmod k}(t) = v | M_u^{e(t)\bmod k}(t) \neq \phi) P(M_u^{e(t)\bmod k}(t+1) \neq \phi, \mathcal{G}_u^v(t+1) | M_u^{e(t)\bmod k}(t) = v) \notag \\
= \sum_{v\in V} P(M_u^{e(t)\bmod k}(t) = v | M_u^{e(t)\bmod k}(t) \neq \phi) P( \mathcal{G}_u^v(t+1) | M_u^{e(t)\bmod k}(t) = v) \notag \\
= \sum_{v\in V} P(M_u^{e(t)\bmod k}(t) = v | M_u^{e(t)\bmod k}(t) \neq \phi) (1-1/n)^{\eta n} \notag  \\
= (1-1/n)^{\eta n}. \notag 
\end{align}
Substituting the above in equation~\eqref{eq:retentsamp1}, we get 
\begin{align}
P(M_u^{e(t)\bmod k}(t+1) \neq \phi) \geq \sigma_{u,t} (1-1/n)^{\eta n}. \notag 
\end{align}
By repeating the above steps, for any $\Delta \in \{1,2,\ldots,k/\eta-1  \}$ we have the lower bound
\begin{align}
P(M_u^{e(t)\bmod k}(t+\Delta) \neq \phi) \geq \sigma_{u,t}(1-1/n)^{\eta n \Delta}. \notag  
\end{align}
Therefore, the expected retention can be given as 
\begin{align}
\rho_{u,t} \geq \sum_{\Delta=0}^{k/\eta-1} \sigma_{u,t} (1-1/n)^{\eta n \Delta}. 
\label{eq:retentsamp2}
\end{align}
For large $n$, equation~\eqref{eq:retentsamp2} can be approximated as 
\begin{align}
\rho_{u,t} \geq \sum_{\Delta=0}^{k/\eta -1} \sigma_{u,t} e^{-\eta \Delta} = \sigma_{u,t} \frac{1 - e^{-k}}{1- e^{-\eta}} \approx \sigma_{u,t}/\eta,   
\end{align}
since $1-e^{-k} \rightarrow 0$ for large $k$ (e.g., $k=\log(n)$) and $e^{-\eta} \approx 1 - \eta$ for small $\eta$ (e.g., $\eta = 0.05$). 
\end{proof}


\subsection{Correctness of Honeybee} \label{analysisModel}
In this section, we theoretically demonstrate that Honeybee nodes achieve near-uniform peer sampling when all nodes follow protocol. 

\smallskip
\noindent 
{\bf Model, assumptions, and notation.} 
For the analysis, we assume that one single node performs random walks according to the \OA~protocol to sample peers and update its address table.
The remaining $V\backslash v$ nodes do not perform random walks, but have oracle access for sampling peers uniformly at random whenever required. 
When a node queries the oracle, it returns an address sampled uniformly at random from $V$. 
Time progresses in discrete rounds $t=0,1,2,\ldots$. 
Each node $v \in V$ has a memory $M_v$ to store addresses of $k$ other nodes in the network. 
Nodes' knowledge about other nodes induces a graph on $V$ which we denote by $G$. 
A snapshot of $v$'s memory and the graph $G$ at time $t$ are denoted as $M_v(t)$ and $G(t)$, respectively. 
At $t=0$, all nodes obtain random sets of addresses from the oracle and $G(0)$ forms a random $k$-regular graph. 
For any two nodes $u, u' \in V$, if $u' \in M_u(t)$ then $u \in M_{u'}(t)$ for all $t\geq 0$. 
We assume all nodes follow the \OA~protocol for this analysis due to the enforcement explained in Appendix~\ref{oa_enforcement}.

During each round, node $v$ does a random walk to update one of the addresses in its address table (memory). 
We let $M_v^i$ be the $i$-th index of the address table of node $v$. 
$M_v^i(t)$ is the $i$-th index of the address table of node $v$ at time $t$. 
In round $t$, node $u$ advances the random walk pertaining to the index \((t\mod k)\) in its address table and attempts to update $M_v^{t\mod k}$. 
A walk is $l$-hops long and can have one of two outcomes: success or failure. 
A successful round is when the terminal node of the walk accepts $v$'s request to mutually add each other's addresses within their respective address tables. 
We assume a walk is successful if the terminal node in the walk is not already in $v$'s address table. 
A successful walk results in an update of $M_v^{t \mod k}$ at the end of round $t$. 
In other words, $M_v^{t\mod k}(t) \neq M_v^{t\mod k}(t+1)$. 
If $u_1, u_2, \ldots, u_l$ is the path taken by $v$ on a successful random walk, with $u_1 \in M_v(t)$ and $u_l \notin M_v(t)$, then $u_l$ is removed from $M_{u_{l-1}}$ (and, $u_{l-1}$ is removed from $M_{u_l}$).
Similarly, $u_1$ is removed from $M_v$ while $v$ is removed from $M_{u_1}$. 
To compensate for the lost addresses from their address tables, we assume the nodes $u_1$ and $u_{l-1}$ connect with each other 
and replenish their tables. 
A failed round is when the random walk terminates at a node that is already in $v$'s address table $M_v(t)$. 
When a random walk fails, the address table does not get updated (i.e.,  $M_v(t) = M_v(t+1)$). 

In Fig. \ref{fig:analysisWalkModel}, we show a toy example of a random walk under the random walk model described above. At the end of the random walk, the initiator node sends a request to the destination node to add each other to their address tables. In other words, at the end of the random walk, the initiator node samples a random node from the penultimate node's address table. 
\begin{figure}[!tbp]
  \centering
  {\includegraphics[width=0.4\textwidth]{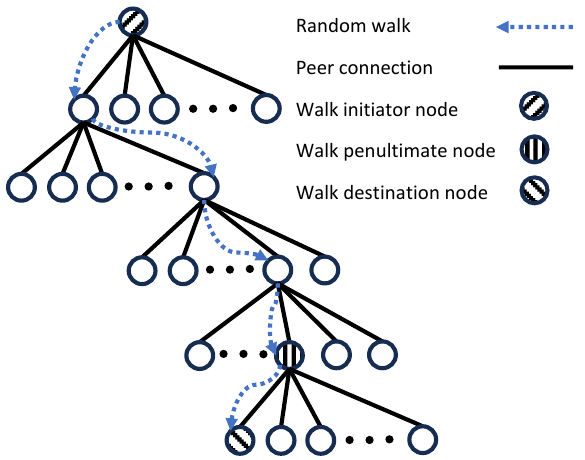}}
  \caption{An example random walk under the model described in \S\ref{analysisModel}. The figure is for illustrative purposes. Some nodes irrelevant to the random walk path are omitted. \OA~is implemented using iterative random walks.}
  \label{fig:analysisWalkModel}
\end{figure}






\smallskip 
\noindent 
{\bf Analysis.}
We define the following two properties for the network.
Our main result shows that the network satisfies these properties for all time rounds $t > 0$.  
\begin{property}\label{property1}
For any two nodes $u, u'$, $M_u(t)$ and $M_{u'}(t)$ are independent and near-uniformly distributed. 
\end{property}
 Property~\ref{property1} says that the address table of any node $u$ is independent of the address table of all nodes and near-uniformly distributed. 
Based on our assumption that all nodes have random addresses sampled from the uniform distribution at round $t=0$, Property~\ref{property1} is trivially satisfied at round $t = 0$. 
\begin{property}\label{property2}
For node $v$, the address table $M_v(t)$ is independent of $M_v(t-k))$ and is near-uniformly distributed. 
\end{property}
We also claim that the address table of node $v$ is independent of $v$'s address table $k$ rounds ago, and is near-uniformly distributed. 
By showing the correctness of properties~\ref{property1} and~\ref{property2}, by induction we can conclude that $v$ generates near-uniform samples that are independent of past samples of $v$. 
\begin{theorem} \label{thm:analysismainresult}
For any time $t$, property 1 and property 2 for node $v$ are satisfied with high probability.   
\end{theorem}

For a random graph, it is known that the local subgraph around a node is tree-like \cite{dreier2018local}. 
Therefore, we consider node $v$'s random walk on the \OA~network as a random walk on a tree of degree $k$. 
Let $R_{u, t, i}$ be the sequence of nodes visited in the random walk initiated by node $u$ pertaining to the \(i\)-th peer in its address table at round $t$. 
For a random walk, there are \(k^l\) total potential outcomes. 
Suppose \(R_{G,l,k}(u)\) is the event that a random walk of length \(l\) starting at node \(u\) returns to \(u\) as its destination. 
We observe that the probability that a random walk of length \(l\) starting at node \(u\) returns to \(u\) as its destination decreases exponentially with the increase of the random walk length. Since \(G\) is a regular tree, the probability of \(R_{G,l,k}(u)\) arrives at the neighbors of \(u\) is bounded from above by \(P(R_{G,l,k}(u))\). Based on our definitions of successful random walk rounds and failed random walk rounds, we formalize our observations into the following lemma.


\begin{lemma}\label{returnProbLemma}
For a random walk of length \(l=2\lambda\), the probability that it fails is \(O \left( \frac{4^\lambda}{k^\lambda \lambda^{\frac{3}{2}}} \right)\).
\end{lemma}
\begin{proof}
    In the regular tree model mentioned before, we can calculate the probability that a random walk of length \(l\) starting at node \(u\) returns to \(u\) as its destination \(P(R_{G,l,k}(u))\) with the potential outcomes for the event \(R_{G,l,k}(u)\), \(n(R_{G,l,k}(u))\), and the total number of potential outcomes, \(k^l\), as follows. 
\begin{align}
P(R_{G,l,k}(u)) &= \frac{n(R_{G,l,k}(u))}{k^l} \label{overall_prob}
\end{align}
A random walk of length \(l\) starting at node \(u\) can return to \(u\) if and only if it consists of \(\frac{l}{2}\) hops outward and \(\frac{l}{2}\) hops inward, and \(l\) is even. The numerator of the right side of Equation (\ref{overall_prob}) can be calculated as follows.
\begin{align}
n(R_{G,l,k}(u)) &= \left[\binom{l}{\frac{l}{2}} - \binom{l}{\frac{l}{2}+1}\right] k (k-1)^{\frac{l}{2}-1}\\
&= \left[\binom{2\lambda}{\lambda} - \binom{2\lambda}{\lambda+1}\right] k (k-1)^{\lambda-1} \label{catalan}
\end{align}
The Catalan number \(C_\lambda=\binom{2\lambda}{\lambda} - \binom{2\lambda}{\lambda+1}\) in Equation (\ref{catalan}) can be expressed as follows.
\begin{align}
C_n &= \binom{2\lambda}{\lambda} - \binom{2\lambda}{\lambda+1}\\
&= \frac{1}{\lambda+1}\binom{2\lambda}{\lambda}\\
&= \frac{1}{\lambda+1} \cdot \frac{(2\lambda)!}{(\lambda!)^2} \label{before_sterling}
\end{align}
Using Stirling's approximation, Equation (\ref{before_sterling}) can be expressed as follows. 
\begin{align}
C_\lambda &= \frac{1}{\lambda+1} \cdot \frac{\sqrt{4\pi \lambda} (\frac{2\lambda}{e})^{2\lambda}(1+O(\frac{1}{\lambda}))}{(\sqrt{2\pi \lambda} (\frac{\lambda}{e})^{\lambda}(1+O(\frac{1}{\lambda})))^2}\\
&= \frac{1}{\lambda(1+O(\frac{1}{\lambda}))} \cdot \frac{4^\lambda}{\sqrt{\pi \lambda}} \cdot \frac{1}{1+O(\frac{1}{\lambda})}\\
&= \frac{4^\lambda}{\sqrt{\pi}\lambda^{\frac{3}{2}}}\left(1+O\left(\frac{1}{\lambda}\right)\right)
\end{align}
Then, by taking \(C_\lambda\) back to Equation (\ref{catalan}), we have the following.
\begin{align}
n(R_{G,l,k}(u)) &= \frac{4^\lambda}{\sqrt{\pi}\lambda^{\frac{3}{2}}}\left(1+O\left(\frac{1}{\lambda}\right)\right) k (k-1)^{\lambda-1}
\end{align}
Hence, we have the probability \(P(R_{G,l,k}(u))\) as follows.
\begin{align}
P(R_{G,l,k}(u)) &= \frac{\frac{4^\lambda}{\sqrt{\pi}\lambda^{\frac{3}{2}}}(1+O(\frac{1}{\lambda})) k (k-1)^{\lambda-1}}{k^{2\lambda}}\\
&= \frac{k (k-1)^{\lambda-1} 4^\lambda(1+O(\frac{1}{\lambda}))}{k^{2\lambda}\sqrt{\pi}\lambda^{\frac{3}{2}}}\\
&= O \left( \frac{k^\lambda 4^\lambda}{k^{2\lambda}\sqrt{\pi}\lambda^{\frac{3}{2}}} \right)\\
&= O \left( \frac{4^\lambda}{k^{\lambda} \lambda^{\frac{3}{2}}} \right)
\end{align}
Specifically, in the \OA~simulation and evaluation, each node stores 24 distinct addresses in its address table, and we have \(k=24\). Thus, in the \OA~network regular tree, \(P(R_{G,l,k}(u))\) can be expressed as follows. 
\begin{align}
P(R_{G,l,k}(u)) &= O \left( \frac{4^\lambda}{24^{\lambda} \lambda^{\frac{3}{2}}} \right)\\
&\approx O \left( \frac{1}{4^{1.2925\lambda} \lambda^{\frac{3}{2}}} \right)
\end{align}
For our analysis, we do not replace \(k\) with a specific number. Therefore, based on our definition of a successful random walk and a failed random walk, the probability that a round fails is \(O \left( \frac{4^\lambda}{k^\lambda \lambda^{\frac{3}{2}}} \right)\), and the probability that a round succeeds is \(1 -  O \left( \frac{4^\lambda}{k^\lambda \lambda^{\frac{3}{2}}} \right)\). In addition, note that these bounds are for situations in which \(l\) is even. \(l\) can be odd in \OA~random walks, and when \(l\) is odd, the random walk does not fail (as discussed in the second paragraph of this proof), further increasing the probability of a round succeeding. 
\end{proof}


We are not ready to prove Theorem~\ref{thm:analysismainresult}. 
We first show that if property~\ref{property1} is true for time $t$, then property~\ref{property2} is true for time $t+k$ with high probability. 
We formalize our observation into the following lemma. 
\begin{lemma}\label{theorem1}
If property~\ref{property1} is true at time $t$, then $P(M_v(t+k) | M_v(t))$ is uniformly distributed with probability \(\left(1 -  O \left( \frac{4^\lambda}{k^\lambda \lambda^{\frac{3}{2}}} \right)\right)^k\).
\end{lemma}
\begin{proof}
The lemma holds true because, as shown in the proof of lemma~\ref{returnProbLemma}, 
 the random walk of a single round has a success probability of \(1 -  O \left( \frac{4^\lambda}{k^\lambda \lambda^{\frac{3}{2}}} \right)\). With \(k\) rounds of random walks, the probability that none of them fails is simply \(\left(1 -  O \left( \frac{4^\lambda}{k^\lambda \lambda^{\frac{3}{2}}} \right)\right)^k\). Since none of the random walks fails, the address table of the random walk initiator node is uniformly distributed at round $t+k$, which means Property \ref{property2} holds with probability that none of the random walks fails.
\end{proof}

Therefore, 
 after \(k\) random walk rounds, node \(v\) still has an address table that is uniformly distributed with a high probability. 
Next, we show that after \(k\) random walk rounds, the address tables of other nodes are uniformly distributed as well. 
\begin{lemma}\label{theorem2}
If property~\ref{property1} holds at time $t-k$ and property~\ref{property2} holds at time $t$, then property~\ref{property1} holds at time $t$. 
\end{lemma}
\begin{proof}
$P(M_u(t+k)| M_v(t+k))$ can be expressed as follows.
\begin{align}
\sum_{a} P(M_v(t)) P(M_v(t+k) | M_v(t)) P(M_u(t+k) | M_v(t+k), M_v(t) = a)\label{theorem6Eq}
\end{align}

The first part in Summation (\ref{theorem6Eq}), $P(M_v(t))$ is uniformly distributed by assumption. 
We have already shown that the second part, $P(M_v(t+k) | M_v(t)$, is uniformly distributed. 
The third part, $P(M_u(t+k) | M_v(t+k), M_v(t))$, involves two kinds of nodes for $u$: (1) nodes that are not affected by $v$'s random walks, and (2) nodes that are affected by $v$'s random walks (i.e., nodes that have their address tables changed due to $v$'s random walks). 
For the first kind of nodes, $P(M_u(t+k)  | M_v(t+k), M_v(t))$ is uniformly distributed because $u$'s address table stays the same. For the second kind of nodes, $P(M_u(t+k)  | M_v(t+k) , M_v(t))$ is still uniformly distributed since $u$ queries the oracle immediately for an address sampled uniformly at random from $V$. Thus, lemma~\ref{theorem2} holds true.
\end{proof}



To conclude, given the address tables of all nodes at round \(t=0\) are uniformly distributed, according to lemma~\ref{theorem2}, Property \ref{property1} is satisfied at all rounds \(t\). 
Since Property \ref{property1} is satisfied at all rounds \(t\), each address sampled by the random walk initiator node is randomly sampled from the uniform distribution, and Property~\ref{property2} is satisfied as well. 
Thus, $P(M_u(t) | M_v(t))$ and $P(M_v(t+k)  | M_v(t))$ are uniformly distributed for all \(t\geq0\) and \(k>0\).
This concludes the proof sketch of Theorem~\ref{thm:analysismainresult}. 

Theorem~\ref{thm:analysismainresult} demonstrates the effectiveness of a single random walk initiated by a node $v$. Assume a \OA~random walk is atomic. That is, at any time $t$, there is only one random walk in graph $G$, and a new random walk starts only after the existing random walk ends. By induction, we can conclude that, for arbitrary positive integers $p$ and $q \geq p$, when $p$ nodes perform a total of $q$ random walks in $G$, they generate $q$ near-uniform samples that are independent of their past samples. 

\end{document}